\documentclass[12pt,oneside,envcountsame,envcountsectt,]{llncs}
\pdfoutput=1

\usepackage{fullpage}

\usepackage{amsmath}
\usepackage{amsfonts}
\usepackage{amssymb}

\usepackage{stmaryrd}
\SetSymbolFont{stmry}{bold}{U}{stmry}{m}{n}
\usepackage{mathrsfs}
\usepackage{subcaption}
\usepackage{enumerate}

\hfuzz=\maxdimen
\usepackage{tikz}
\usetikzlibrary{positioning}
\usetikzlibrary{calc}
\tikzstyle{proc}=[circle, draw, inner sep=0pt,
                  minimum width=4pt, line width=0.1pt]
\tikzstyle{proc1}=[circle, draw, inner sep=0pt,
                  minimum width=3pt, line width=0.1pt]
\tikzstyle{proc2}=[circle, draw, inner sep=0pt,
                  minimum width=2pt, line width=0.1pt]
\tikzstyle{n}=[fill=black]
\tikzstyle{b}=[fill=white]
\tikzstyle{g}=[fill=gray]
\usepackage{standalone}
\usepackage{graphics, color} %
\usepackage{epsfig}
\usepackage[boxruled,vlined,linesnumbered]{algorithm2e}
\usepackage{alltt}

\let\oldnl\nl%
\newcommand{\nonl}{\renewcommand{\nl}{\let\nl\oldnl}}%

\makeatletter
\renewcommand{\algocf@caption@boxruled}{%
  \hrule
  \hbox to \hsize{%
    \vrule\hskip-0.4pt
    \vbox{
       \vskip\interspacetitleboxruled%
       \unhbox\algocf@capbox\hfill
       \vskip\interspacetitleboxruled
       }%
     \hskip-0.4pt\vrule%
   }\nointerlineskip%
}
\makeatother

\usepackage{xspace} %
\newcommand{\macro}[2]{ \providecommand{#1}{{{\ensuremath{#2}}}\xspace}}
\newcommand{\rmacro}[2]{ \renewcommand{#1}{{{\ensuremath{#2}}}\xspace}}

\macro{\N}{\mathbb N}
\macro{\Z}{\mathbb Z}
\macro{\R}{\mathbb R}

\macro{\noir}{\bullet}
\macro{\blanc}{\circ}
\macro{\lnoir}{\blanc\!\!\rightarrow\!\!\noir}
\macro{\lblanc}{\blanc\!\!\leftarrow\!\!\noir}
\macro{\lall}{\blanc\!\!-\!\!\!\!-\!\;\!\noir}
\macro{\lok}{\blanc\!\!\leftrightarrow\!\!\noir}
\macro{\proc}{p}%
\macro{\lproc}{\odot}
\macro{\Tnoir}{\{\lok, \lnoir\}^\omega}
\macro{\Tblanc}{\{\lok, \lblanc\}^\omega}
\macro{\Szero}{\{\lok\}^\omega}
\macro{\Cun}{\{\lok\}^\omega \cup \{\lok\}^*\{\lblanc\}^\omega \cup \{\lok\}^*\{\lnoir\}^\omega}
\macro{\Sun}{\{\lok, \lblanc\}^\omega \cup \{\lok, \lnoir\}^\omega}
\macro{\Run}{\{\lok, \lblanc, \lnoir\}^\omega}
\macro{\algo}{\mathcal{A}}
\macro{\SymbSimplex}{\mathcal{S}}
\macro{\emptyset}{\o}
\macro{\coloneqq}{:=}
\macro{\compos}{.}
\macro{\cr}{cr}
\newcommand{\ie}{\emph{i.e.} }

\macro{\TS}{\Phi}  %
\rmacro{\b}{\circ}
\macro{\n}{\bullet}
\macro{\g}{\textcolor{gray}{\bullet}}
\macro{\Go}{\Gamma^\omega}
\macro{\I}{\mathcal{I}}
\newcommand{\eg}{\emph{e.g.} }
\macro{\F}{\mathcal F}
\renewcommand{\O}{\mathcal{O}}
\rmacro{\L}{\mathcal L}

\newcounter{theoenumcounter}
\def\theoenumlabel{\thetheorem.\roman{theoenumcounter}}

\newenvironment{theoenum}{\begin{list}{\theoenumlabel}{%
      \usecounter{theoenumcounter}
}      %
  }{\end{list}}

\begin{document}
\pagestyle{plain}
\thispagestyle{empty}

\title{Back to the Coordinated Attack Problem}

\author{
  Emmanuel Godard \and Eloi Perdereau\\
\texttt{(emmanuel.godard|eloi.perdereau)@lis-lab.fr}
}
\institute{Laboratoire d'Informatique et Syst{\`e}mes\\
Aix-Marseille Universit{\'e} -- CNRS (UMR 7020)}
\date{ }
\pagestyle{plain}
\maketitle

\begin{abstract}
We consider the well known Coordinated Attack Problem, where two
generals have to decide on a common attack, when their messengers can
be captured by the enemy.
Informally, this  problem represents the difficulties to agree in the
presence of communication faults.
We consider here only omission faults
(loss of message), but contrary to previous studies, we do not to
restrict the way messages can be lost, \ie we make no specific assumption,
we use no specific failure metric.

In the large subclass of message
adversaries where the %
double simultaneous omission can
never happen, we characterize which ones are obstructions for the
Coordinated Attack Problem.
We give two proofs of this result. One is combinatorial and uses the
classical bivalency technique for the necessary condition.
The second is topological and uses simplicial complexes to prove the necessary condition. We also present two different Consensus algorithms that are combinatorial (resp. topological) in essence.
Finally, we  analyze the two
proofs and illustrate the relationship between the combinatorial approach and
the topological approach in the very general case of message
adversaries. 
We show that the topological characterization gives a clearer explanation of why some message adversaries are obstructions or not.
This result is a convincing illustration of the power
of topological tools for distributed computability.
\end{abstract}


\section{Introduction}

\subsection{Motivations}

The  Coordinated  Attack Problem (also known in the literature as
the two generals or the two armies problem) is a long time problem in
the area of distributed computing.
It is a fictitious situation where two armies have to agree to attack
or not on a common enemy that is between them and might capture any of
their messengers.
Informally, it represents the  difficulties to agree in the
presence of communication faults. The design of a solution is
difficult, sometimes impossible, as it has
to address possibly an infinity of lost mutual acknowledgments.
It has important applications for the Distributed Databases commit for
two processes, see \cite{Gray78}. It was one of the first
impossibility results in the area of fault tolerance and distributed
computing \cite{akkoyunlu_constraints_1975,Gray78}.

In the vocabulary of more recent years, this problem can now be stated
as the Uniform Consensus Problem for 2 synchronous processes communicating by
message passing in the presence of omission
faults. It is then a simple instance of a problem that had been very widely
studied \cite{AT99,CBGS00,MR98}.
 See for example
\cite{RaynalSynchCons} for a
  recent survey about Consensus on synchronous systems with some
  emphasis on the omissions fault model.

Moreover, given that, if any message can be lost, the impossibility of
reaching an agreement is obvious, one can wonder why it has such
importance to get a name on its own, and maybe why it has been studied in
the first place...
The idea is that one has usually to restrict the way the
messages are lost in order to keep this problem relevant.

We call an arbitrary pattern of
failure by loss of messages a \emph{message adversary}. It will
formally describe the fault environment in which the system evolves.
For example, the message adversary where any message can be lost at any
round except that all messages cannot be lost indefinitely is a
special message adversary (any possibility of failure \emph{except one}
scenario), for which it is still impossible to solve the
Coordinated Attack Problem, but the proof might be less trivial.
Given a message adversary, a natural question is whether the
Coordinated Attack Problem is solvable against this environment. More
generally, the question that arises now is to describe
what are exactly the message adversaries for which the
Coordinated Attack Problem admit a solution, and for which ones is there no
solution. These later message adversaries will be called \emph{obstructions}.

\subsection{Related Works}

The Coordinated Attack Problem is a kind of folklore problem for
distributed systems.
It seems to appear first in \cite{akkoyunlu_constraints_1975} where it is a
problem of gangsters plotting for a big job. It is usually attributed to
\cite{Gray78}, where Jim Gray coined the name ``Two Generals Paradox''
and put the emphasis on the infinite recursive need for
acknowledgments in the impossibility proof.

In textbooks it is often given as an example, however the drastic
conditions under which this impossibility result yields are never
really discussed, even though for relevancy purpose they are often
slightly modified.
In \cite{LynchDA}, a different problem of Consensus (with a
weaker validity condition) is used.
In \cite{DADA}, such a possibility of eternal loss is
explicitly ruled out as it would give a trivial impossibility proof
otherwise.

This shows that the way the messages {may} be lost is an important
part of the problem definition,
hence it is interesting to characterize when the pattern of
loss allows to solve the consensus problem or not, \ie
whether the fault environment is an obstruction for the Coordinated
Attack Problem or not.
To our knowledge, this is the first time, this problem is investigated
for arbitrary  patterns of omission failures, even in the simple case
of only two processes.
Most notably, it has been addressed for an arbitrary number
of processes and for special (quite regular) patterns in
\cite{CHLT00,GKP03,RaynalSynchCons}.

A message adversary is oblivious if the set of possible communication
pattern is the same at each step. The complete characterization of
oblivious message adversaries for which Consensus is solvable has been
given in \cite{CGPcarac}. We consider here also non oblivious
adversaries.

Note that while the model has been around for at least tens of years
\cite{timeisnotahealer}, the name ``message adversary'' has been
coined only recently by Afek and Gafni in \cite{messadv}.

\subsection{Scope of Application and Contributions}

Impossibility results in distributed computing are the more
interesting when they are tight, \ie when they give an exact
characterization of when a problem is solvable and when it is not.
There are a lot of results regarding distributed tasks
when the underlying network is a complete graph and the patterns are simply described by faults (namely in the context
of Shared Memory systems), see for example the works in
\cite{HS99,SZ,BoGa93} where exact topologically-based
characterizations is given for the wait-free model.

There are also more recent results  when the underlying
network can be any arbitrary graph. The results given in \cite{SW07} by
Santoro and Widmayer are almost tight. What is worth to note is that
the general theorems of \cite{HS99} could not be directly used by this study
for the very reason that the failure model for communication networks
is not interestingly expressible in the fault model for systems with
one to one communication.
In the following of~\cite{CS09}, we are not interested in the exact
cause of a message not being sent/received. We are interested in as general as possible models. See Section~\ref{faultmetric} for a more detailed discussion.

We underline that the omission failures we
are studying here encompass networks with \emph{crash failures}, see
Example~\ref{ex:crash}.
It should also be clear that message adversaries can also be studied
in the context of a problem that is not the Consensus Problem.
Moreover, as we do not endorse any pattern of failures as being, say,
more ``realistic'', our technique can be applied for any
new patterns of failures.
\medskip

On the way to a thorough
characterization of all obstructions, we
address the Consensus problem for a particular but important
subclass of failure patterns,
namely the ones when no two messages can be lost at the same round.
It is long known that the Coordinated Attack Problem is unsolvable if at most one
message might be lost at each round \cite{CHLT00,GKP03}. But what happens with strictly weaker patterns was unknown.

Our contribution is the following. %
In the large subclass of message adversaries where the %
double simultaneous omission can never happen, we characterize which
ones are obstructions for the Coordinated Attack Problem.  We give two
alternative proofs of this result. One is based on traditional
combinatorial techniques, that have been extended to the more involved
setting of arbitrary message adversaries. The second one presents an
extension of topological techniques suited to arbitrary message adversaries.
The combinatorial proof was
presented in \cite{FG11}. The topological proof is an improved extract
from the Master thesis of one of the authors \cite{eloiM2} where the notion of terminating subdivision from \cite{GKM14} is applied to the setting of the Coordinated Attack Problem.

More interestingly, the topological characterization gives a nice
topological unified explanation of the characterization, which is
separated in four different cases in the combinatorial presentation
from \cite{FG11}.  This result is a convincing illustration of the
power of topological tools for distributed computability. Topological
tools for distributed computing have sometimes been criticized for being
over-mathematically involved. The result presented in this
paper shows that distributed computability is inherently linked to
topological properties.

But the paper also illustrate some pitfalls of such tools as we could
use the given characterization to uncover an error in the main theorem
of a paper generalizing the Asynchronous Computability Theorem to
arbitrary message adversaries \cite{GKM14}. Mathematically, this error
can be traced to the subtle differences between simplicial complexes
and abstract simplicial complexes of infinite size.
See Remark~\ref{subtle}

\bigskip
The outline of the paper is the following.
We describe Models and
define our Problem in the Section~\ref{def}. We present numerous
examples of application of our terminology and notation in
Section~\ref{examples}. We then address the characterization of
message adversaries without simultaneous faults that are obstructions for
the Coordinated Attack Problem in Theorem~\ref{thm:FG11}.
We give the proofs for necessary condition
(impossibility result) in \ref{CN} and for sufficient condition
(explicit algorithm) in \ref{CS} using the classical bivalency techniques.
We then prove the same results in a topological way in section \ref{sec:topo}.
Finally, the two results are compared and we show how the topological explanation gives more intuition about the result.

\section{Models and Definitions}

\label{def}

 \subsection{The Coordinated Attack Problem}

\subsubsection{A folklore problem}
Two generals have gathered forces on top of two facing hills. In between, in the
valley, their common enemy is entrenched. Every day each general
sends a messenger to the other through the valley. However this
is risky as the enemy may capture them. Now they need to get
the last piece of information: are they \emph{both} ready to attack?

This two army problem was originated
by~\cite{akkoyunlu_constraints_1975} and then by Gray~\cite{Gray78}
when modeling the distributed database commit.
It corresponds to the binary consensus with
two processes in the omission model. If their is no restriction on the
fault environments, then any messenger may be captured.
And if any messenger may be captured, then consensus is obviously
impossible: the enemy can succeed in capturing all of them, and
without communication, no distributed algorithm.

\subsubsection{Possible Environments}

Before trying to address what can be, in some sense, the most relevant
environments, we will describe
different environments in which the enemy cannot be
so powerful as to be able to capture any messenger.

This is a list of possible environments, using the same military
analogy. The generals name are \emph{White} and \emph{Black}.
\label{7cases}

\begin{enumerate}[(1)]
\item
no messenger is captured
\item
messengers from \emph{General White} may be captured
\item
messengers from \emph{General Black} may be captured
\item
messengers from one general are at risk, and if one of them is captured,
all the following will also be captured (the enemy got the secret ``Code
of Operations'' for this general from the first captured messenger)
\item
messengers from one general are at risk (the enemy could manage to infiltrate
a spy in one of the armies) %
\item
at most one messenger may be captured each day (the enemy can't closely
watch both armies on the same day)
\item
any messenger may be captured
\end{enumerate}

Which ones are (trivial) obstructions, and which are not? Nor
obstruction, nor trivial? What about more complicated environments?

\subsection{The Binary Consensus Problem}
\label{defconsensus}
A set of synchronous processes wish to agree about a binary
value. This problem was first identified and formalized by Lamport,
Shostak and Pease \cite{LSP}. Given a set of processes, a consensus
protocol must satisfy the following
properties for any combination of initial values~\cite{LynchDA}:

\begin{itemize}
\item \emph{Termination}: every process decides some value.
\item \emph{Validity}: if all processes initially propose the same value $v$,
then every process decides $v$.
\item \emph{Agreement}: if a process decides $v$, then every process decides
 $v$.
\end{itemize}

Consensus with such a termination and decision requirement for every
process is more precisely referred to as the \emph{Uniform Consensus},
see \cite{RaynalSynchCons} for example for a discussion.
Given a fault environment, the natural questions are : is the
Consensus solvable, if it is solvable, what is the minimal complexity?

\subsection{The Communication Model}
In this paper,
the system we consider is a set $\Pi$ of only $2$ processes named \emph{white}
and \emph{black},  $\Pi = \{\blanc,\noir\}$.
The processes evolve in \emph{synchronized} rounds.
In each round $r,$ every process \proc executes the following
steps: \proc sends a message to the other process, receives a
messages $M$ from the other process, and then updates its state
according to the received message.

The messages sent are, or are not, delivered according to the
environment in which the distributed computation takes place. This
environment is described by a \emph{message adversary}. Such an adversary
 models exactly when some
messages sent by a process to the other process may be lost
at some rounds.
We will consider arbitrary message adversaries, they will be represented
as  arbitrary sets of infinite sequences of combinations of communication rounds. 

In the following of~\cite{CS09}, we are not interested in the exact
cause of a message not being sent/received. We only refer to the phenomenon:
the content of messages being transmitted or not.
The fact that the adversaries are arbitrary means that we do not endorse
any metric to count the number of ``failures'' in the
system. There are metrics that count the number of lost messages, that
count the number of process that can lose messages (both in send and
receive actions). Other metrics count the same parameters but only
during a round of the system. The inconvenient of this metric-centric
approach  is that, even when restricted only to omission faults,  it
can happen that some results obtained
on complete networks are not usable on arbitrary networks. Because in,
say, a ring
network, you are stating in some sense that every node is
faulty. See eg \cite{SW07} for a very similar discussion,  where
Santoro and Widmayer, trying to solve some generalization of
agreement problems in general networks could not use directly the
known results in complete networks.
\label{faultmetric}

As said in the introduction, the Coordinated Attack Problem is
nowadays stated as the Uniform Consensus Problem for 2 synchronous
processes communicating by message passing in the presence of omission
faults. Nonetheless, we will use Consensus and Uniform Consensus
interchangeably in this paper. We emphasize that, because we do not
assign omission faults to any process (see previous discussion), there
are only correct processes.

\subsection{Message Adversaries}
\label{subsec:def}
We introduce and present here our notation.

\begin{definition}
  We denote by $\mathcal G_2$  the set of directed graphs with vertices
  in $\Pi$.
  $$\mathcal G_2=\{\lok,\lblanc,\lnoir,\lall\}$$
    We denote by $\Gamma$  the following subset of $\mathcal G_2$
$$\Gamma=\{\lok,\lblanc,\lnoir\}.$$
\end{definition}

At a given round, there are only four combinations of communication.
Those elements describe what can happen at a given round with the
following straightforward semantics:

\begin{itemize}
\item
$\lok,$ no process looses messages
\item
$\lblanc,$ the message of process \blanc, if any, is not transmitted
\item
$\lnoir,$ the message of process \noir, if any, is not transmitted
\item
$\lall,$ both messages, if any, are not transmitted\footnote{In order
  to increase the readability, we note \lall
  instead of $\blanc\;\;\noir$, the double-omission case.}.
\end{itemize}

The terminology \emph{message adversary} has been introduced in \cite{messadv}, but the concept is way older.
\begin{definition}
  A \emph{message adversary} over $\Pi$ is a set of infinite
  sequences of elements of $\mathcal G_2$.
\end{definition}

We will use the standard following notations in order to describe more easily
our message adversaries \cite{PPinfinite}. A (infinite) sequence is
seen as a (infinite) word over the alphabet $\mathcal G_2$.
The empty word is denoted by $\varepsilon$.

\begin{definition}
  Given $A\subset\mathcal G_2$, $A^*$ is the set of all finite
  sequences of elements of $A$, $A^\omega$ is the set
  of all infinite ones and $A^\infty = A^* \cup A^\omega.$
\end{definition}

An adversary of the form $A^\omega$ is called an \emph{oblivious adversary}.
A word in $L\subset\mathcal G_2^\omega$ is called a \emph{communication
  scenario} (or \emph{scenario} for short) of message adversary $L$.
Given a word $w\in\mathcal G_2^*$, it is called a \emph{partial scenario} and
$len(w)$ is the length of this word.

Intuitively, the letter at position $r$ of the word describes whether
there will be, or not, transmission of the message, if one is
being sent at round $r$.
A formal definition of an execution under a scenario will be given in
Section~\ref{execution}.

\medskip

We recall now the definition of the prefix of words and languages.
A word $u\in\mathcal G_2^*$ is a prefix for $w\in\mathcal G_2^*$
(resp. $w'\in\mathcal G_2^\omega$) if there exist $v\in\mathcal G_2^*$
(resp. $v'\in\mathcal G_2^\omega$) such that $w=uv$ (resp. $w'=uv'$).

Given
$w\in\mathcal G_2^\omega$ and $r\in\N$, $w_{|r}$ is the prefix of size $r$ of $w$.

\begin{definition}
  Let $w\in\mathcal G_2^*$, then $Pref(w)=\{u\in\mathcal G_2^*|u \mbox{ is a
    prefix of } w\}$. Let $L\subset\mathcal G_2^*$, let $r\in\N$,
  $Pref_r(L)=\{w_{|r}\mid w\in L\}$ and
  $Pref(L)=\mathop\bigcup\limits_{w\in L}Pref(w)=\mathop\bigcup\limits_{r\in \N}Pref_r(L)$.
\end{definition}

\subsection{Examples}
\label{examples}
  We do not restrict our study to regular languages, however all message
  adversaries we are aware of are regular, as can be seen in the
  following examples, where the rational expressions prove to be very
  convenient.

\medskip
We show how standard fault environments are conveniently described in
our framework. %

\begin{example}
  Consider a system where, at each round, up to $2$ messages can be
  lost. The associated message adversary is $\mathcal G_2^\omega$.
\end{example}

\begin{example}
  Consider a system where, at each round, only one message can be
  lost. The associated message adversary is
  $\{\lok,\lblanc,\lnoir\}^\omega=\Gamma^\omega$.
\end{example}

\begin{example}
  Consider a system where at most one of the processes can lose
  messages. The
  associated adversary is the following:
  $$S_1 = \Sun$$

\end{example}

\begin{example}\label{ex:crash}
  Consider a system where at most one of the processes can crash,
  For the phenomena point of view, this is equivalent to the fact that
  at some point, no message from a particular process will be
  transmitted. The associated adversary is the following:

  $$C_1=\{\lok^\omega\}\cup\{\lok\}^*(\{\lblanc^\omega, \lnoir^\omega\})$$

\end{example}

\begin{example} \label{ex7cases}
  Finally the seven simple cases exposed by the possible environments described in Section \ref{7cases} are described
  formally as follows:
\begin{eqnarray}
\label{S_0}
S_0&=&\Szero \\
\label{T_blanc}
T_\blanc&=&\Tblanc \\
\label{T_noir}
T_\noir&=&\Tnoir \\
\label{C_1}
C_1&=&\{\lok^\omega\}\cup\{\lok\}^*(\{\lblanc^\omega, \lnoir^\omega\})\\
\label{S_1}
S_1&=&\Sun = T_\blanc \cup T_\noir \\
\label{R_1}
R_1&=&\Go \\%
\label{S_2}
S_2&=&\mathcal G_2^\omega
\end{eqnarray}

Note that the fourth and first cases correspond respectively to the
synchronous crash-prone model~\cite{LynchDA} and to the
$1-$resilient model~\cite{Adagio03,GKP03}. Even though in our definition, sets of possible scenarios could be arbitrary,
it seems that all standard models (including crash-based models) can be described using only regular expressions. 
\end{example}

\subsection{Execution of a Distributed Algorithm}
\label{execution}
Given a message adversary $L$, we define  what is a run of
a given algorithm \algo subject to $L$.

An execution, or run, of an algorithm  \algo under scenario $w\in L$
is the following. At round $r\in\N$, messages are sent (or not) by the
processes. The fact that the corresponding receive action will be
successful depends on $a$, the $r$-th letter of $w$.

\begin{itemize}
\item if $a = \lok,$ then all messages, if any, are correctly delivered,
\item if $a = \lblanc,$ then the message of process \blanc is not
  transmitted (the receive call of \noir, if any at this round,
  returns $null$),
\item if $a = \lblanc,$ then the message of process \noir is not
  transmitted (the receive call of \blanc, if any at this round,
  returns $null$),
\item if $a = \lall,$ no messages is transmitted.
\end{itemize}

Then, both processes updates their state according to \algo and the
value received.
An execution is a (possibly infinite) sequence of such
messages exchanges and corresponding local states.

\medskip
Given $u\in Pref(w)$, we denote by $s^\proc(u)$ the
state of process \proc at the $len(u)$-th round of the algorithm \algo
under scenario $w$. This means in particular that
$s^\proc(\varepsilon)$ represents the initial state of \proc, where
$\varepsilon$ denotes the empty word.

Finally and classically,
\begin{definition}
An algorithm \algo solves the Coordinated Attacked Problem for the
message adversary $L$ if  for any
scenario $w\in L$, there exist $u\in Pref(w)$ such that the states
of the two processes ($s^\blanc(u)$ and $s^\noir(u)$) satisfy the
three conditions of Section~\ref{defconsensus}.
\end{definition}

\begin{definition}
  A message adversary $L$ is said to be \emph{solvable} if there exist an
  algorithm that solves the Coordinated Attacked Problem for $L$. It
  is said to be an \emph{obstruction} otherwise.

  A message adversary $L$ is a (inclusion) minimal obstruction if any
  $L'\varsubsetneq L$ is solvable.
\end{definition}

\subsection{Index of a Scenario}
\label{sub:ind}
We will use the following integer function of scenarios that will be
proved to be a useful encoding of all the important properties of
a given message adversary by mapping scenarios in $\Gamma^r$ to integers in
$[0, 3^r-1]$. The intuition to this function will become clearer from the
topological point of view in Section~\ref{sec:topo}.
By induction, we define the following integer index  given
$w\in\Gamma^*$. First, we define $\mu$ on $\Gamma$ by
\begin{itemize}
\item $\mu(\lnoir)=-1$,
\item $\mu(\lok)=0$.
\item $\mu(\lblanc)=1$,
\end{itemize}

\begin{definition}
  Let $w\in\Gamma^*$. We define $ind(\varepsilon)=0$.
  If $len(w)\geq1$, then we have $w=ua$ where $u\in\Gamma^*$ and
  $a\in\Gamma$. In this case, we define
  $$ind(w):=3ind(u)+(-1)^{ind(u)}\mu(a)+1.$$
\end{definition}

Let $n\in\N$, define $ind_n\colon\Gamma^n\to[0,1]$ and $\overline{ind}\colon\Go\to[0,1]$ to
be respectively the normalization of $ind$ and the limit index of the infinite
scenarios~:
\begin{itemize}
  \item $\forall w\in\Gamma^n \quad ind_n(w) = \frac{ind(w)}{3^{n}}$
  \item $\forall w\in\Go \quad \overline{ind}(w) = \lim\limits_{n\to+\infty} ind_n(w_{\mid n})$
\end{itemize}

The convergence for $\overline{ind}$ is obvious from the following lemma.

\begin{lemma}\label{bijindex}
  Let $r\in\N$. The application $ind$ is a bijection from $\Gamma^r$
  to $\llbracket 0,3^r-1\rrbracket.$
\end{lemma}
\begin{proof}
  The lemma is proved by a simple induction. If $r=0$, then the
  property holds.

  Let $r>0$. Suppose the property is satisfied for $r-1$. Given
  $w\in\Gamma^r$, we have $w=ua$ with $u\in\Gamma^{r-1}$ and
  $a\in\Gamma$. From $ind(w)=3ind(u)+(-1)^{ind(u)}\mu(a)+1$ and
  the induction hypothesis, we get immediately that $0\leq ind(w)\leq
  3^r - 1$.

  Now, we need only to prove injectivity. Suppose there are
  $w,w'\in\Gamma^r$ such that $ind(w)=ind(w').$ So there are
  $u,u'\in \Gamma^{r-1}$ and $a,a'\in \Gamma$ such that
  $$3ind(u)+(-1)^{ind(u)}\mu(a)+1 =
    3ind(u')+(-1)^{ind(u')}\mu(a')+1.$$
  Then   $$3(ind(u) - ind(u')) = (-1)^{ind(u')}\mu(a') -
  (-1)^{ind(u)}\mu(a).$$
  Remarking that the right hand side of this integer equality has an
  absolute value that can be at most 2, we finally get
  \begin{eqnarray*}
    ind(u) &=& ind(u')\\
    (-1)^{ind(u)}\mu(a) &=& (-1)^{ind(u')}\mu(a')
  \end{eqnarray*}

  By induction hypothesis, we get that $u=u'$ and $a=a'$. Hence
  $w=w'$, and $ind$ is injective, therefore bijective from $\Gamma^r$
  onto $\llbracket 0,3^r-1\rrbracket.$
\end{proof}

Two easy calculations give
\begin{proposition}
  Let $r\in\N$, $ind(\lnoir^r)=0$ and $ind(\lblanc^r)=3^r-1.$
\end{proposition}

In Figure~\ref{fig:index}, the indexes for words of length at most 2
are given.

\begin{figure}
  \centering
  \begin{tabular}[c]{|r|c|c|c|}
    \hline
    word of length $1$&\lnoir&\lok&\lblanc\\
    \hline
    index&0&1&2\\
    \hline
  \end{tabular}

  \medskip

  \begin{tabular}[c]{|r|c|c|c|}
    \hline
    word of length $2$&\lnoir\lnoir&\lnoir\lok&\lnoir\lblanc\\
    \hline
    index&0&1&2\\
    \hline
  \end{tabular}

  \medskip

  \begin{tabular}[c]{|r|c|c|c|}
    \hline
    word of length $2$&\lok\lnoir&\lok\lok&\lok\lblanc\\
    \hline
    index&5&4&3\\
    \hline
  \end{tabular}

  \medskip

  \begin{tabular}[c]{|r|c|c|c|}
    \hline
    word of length $2$&\lblanc\lnoir&\lblanc\lok&\lblanc\lblanc\\
    \hline
    index&6&7&8\\
    \hline
  \end{tabular}
  \caption{Indexes for some short words}
  \label{fig:index}
\end{figure}

We now describe precisely what are the words whose indexes
differs by only 1. We have two cases, either they have the same prefix and different last letter 
or different prefix and same last letter.
\begin{lemma}\label{diff1}
  Let $r\in\N$, and $v,v'\in\Gamma^r$. Then $ind(v')=ind(v)+1$ if
  and only if one of the following conditions holds:
  \begin{theoenum}
  \item $ind(v)$ is even and
    \begin{itemize}
    \item either there exist $u\in\Gamma^{r-1}$, and $v=u\lok$,
      $v'=u\lnoir$,
    \item either there exist $u,u'\in\Gamma^{r-1}$, and $v=u\lblanc$,
      $v'=u'\lblanc$, and $ind(u')=ind(u)+1.$
    \end{itemize}
  \item $ind(v)$ is odd and
    \begin{itemize}
    \item either there exist $u\in\Gamma^{r-1}$, and $v=u\lblanc$,
      $v'=u\lok$,
    \item either there exist $u,u'\in\Gamma^{r-1}$, and $v=u\lnoir$,
      $v'=u'\lnoir$, and $ind(u')=ind(u)+1.$
    \end{itemize}
  \end{theoenum}
\end{lemma}

\begin{proof}
  The lemma is proved by a induction. If $r=0$, then the
  property holds. Let $r\in\N^*$. Suppose the property is satisfied
  for $r-1$.

  Suppose there are
  $w,w'\in\Gamma^r$ such that $ind(w')=ind(w)+1.$ So there are
  $u,u'\in \Gamma^{r-1}$ and $a,a'\in \Gamma$ such that
  $$3ind(u)+(-1)^{ind(u)}\mu(a)+2 =
    3ind(u')+(-1)^{ind(u')}\mu(a')+1.$$

    Then   $$3(ind(u) - ind(u')) = (-1)^{ind(u')}\mu(a') -
  (-1)^{ind(u)}\mu(a) - 1.$$

  Remarking again that this is an integer equality, we then have
  \begin{itemize}
  \item either $(ind(u)=ind(u')$ and $(-1)^{ind(u')}\mu(a') =
    (-1)^{ind(u)}\mu(a) + 1$,
  \item either $(ind(u')=ind(u)+1$ and $(-1)^{ind(u')}\mu(a') =
  (-1)^{ind(u)}\mu(a)-2.$
  \end{itemize}

  This yields the following cases:
  \begin{enumerate}
  \item $ind(u)=ind(u')$ is even, $a=\lblanc$ and $a'=\lok$,
  \item $ind(u)=ind(u')$ is odd, $a=\lnoir$ and $a'=\lok$,
  \item $ind(u)=ind(u')$ is even, $a=\lok$ and $a'=\lnoir$,
  \item $ind(u)=ind(u')$ is odd, $a=\lok$ and $a'=\lblanc$,
  \item $ind(u')=ind(u)+1$, $ind(u)$ is even, $a=\lblanc=a'$,
  \item $ind(u')=ind(u)+1$, $ind(u)$ is odd, $a=\lnoir=a'$,
  \end{enumerate}

  Getting all the pieces together, we get the results of the lemma.
\end{proof}

Given an algorithm \algo, we have this fundamental corollary, that
explicit the uncertainty process can experience between two executions whose
indexes differ by only 1: one of the process is in the same state
in both cases. Which process it is depends on the parity.  In other
word, when the first index is even, \noir cannot distinguish the two
executions, when it is odd, this is \blanc that cannot distinguish the
two executions.

\begin{corollary}\label{state}
  Let $v,v'\in\Gamma^r$ such that $ind(v')=ind(v)+1$. Then,
  \begin{theoenum}
  \item if $ind(v)$ is even then $s^\noir(v)=s^\noir(v')$,
  \item if $ind(v)$ is odd then $s^\blanc(v)=s^\blanc(v')$.
  \end{theoenum}
\end{corollary}

\begin{proof}
  We prove the result using
  Lemma~\ref{diff1} and remarking that either a process receives a
  message from the other process being in the same state in the
  preceding configuration $u$; either it receives no message when the
  state of the other process actually differ.
\end{proof}

\subsection{Characterization}
\label{sub:comb_charac}
We prove that a message adversary $L\subset\Gamma^\omega$ is solvable
if and only if it does not contain a fair scenario or a special pair
of unfair scenarios. We define the following set to help describe the
special unfair pairs.

\begin{definition}
  A scenario $w\in\mathcal G_2^\omega$ is \emph{unfair} if
  $w\in\mathcal G_2^*(\{\lall,\lblanc\}^\omega\cup\{\lall,\lnoir\}^\omega)$.
  The set of fair scenarios of $\Gamma^\omega$ is denoted by
  $Fair(\Gamma^\omega)$.
\end{definition}
In words, in an \emph{unfair} scenario, there is one or more processes for
which the messages are indefinitely lost at some point. And in a \emph{fair}
scenario there is an infinity of messages that are received from both
processes.

\begin{definition}\label{pair} We define \emph{special pairs}
  as $SPair(\Gamma^\omega)=\{(w,w')\in\Gamma^\omega\times\Gamma^\omega
    \mid w\neq w',
    \forall r\in\N |ind(w_{\mid r})-ind(w'_{\mid r})|\leq1\}.$

\end{definition}

\begin{theorem}\label{thm:FG11}
  Let $L\subset\Gamma^\omega$, then Consensus is solvable for message
  adversary $L$ if and only if 
  $L\in\{\F_1\cup\F_2\cup\F_3\cup\F_4\}$ where
  \begin{theoenum}
  \item $\F_1 = \{L\subset\Go \mid \exists f\in Fair(\Go)\wedge f\notin L\}$
  \item $\F_2 = \{L\subset\Go \mid \exists (w,w')\in SPair(\Go)\wedge
                                                              w,w'\notin L\}$
  \item $\F_3 = \{L\subset\Go \mid \lnoir^\omega\notin L\}$
  \item $\F_4 = \{L\subset\Go \mid \lblanc^\omega\notin L\}$
  \end{theoenum}
\end{theorem}
We have split the set of solvable scenarios in four families for a better
understanding even though it is clear that they largely intersect.
$\F_1$ contains every scenarios for which at least one unfair scenario cannot
occur ; in $\F_2$ \emph{both} elements of a special pair cannot occur ; finally
$\F_3$ and $\F_4$ contains every scenarios for which at least one message is
received from both processes.

We present two proofs of Theorem~\ref{thm:FG11} in the following
sections.

\subsection{Application to the Coordinated Attack Problem}
\label{backto}

We consider now our question on the seven examples of
Example~\ref{ex7cases}.

The answer to possibility is obvious for the first and last cases.
In the first three cases, $S_0$, $\Tblanc$ or $\Tnoir$,
Consensus can be reached in one day (by deciding the initial value of \blanc,\noir, and \blanc respectively).
The fourth and fifth
cases are a bit more difficult but within reach of our
Theorem~\ref{thm:FG11}. We remark that the scenario
$\lblanc\lnoir\lok^\omega$ is a fair scenario that does not belong to
$C_1$, nor $S_1$. Therefore, those are solvable cases also.
 In the last
case, consensus can't be achieved ~\cite{Gray78,LynchDA}, as said before.

The following observation is also a way to derive lower bounds from computability results.

\begin{proposition}
  \label{complexity}
  Let $O\subset \Gamma^\omega$ an obstruction for the Consensus problem. Let $L\subset \Gamma^\omega$ and $r\in\N$ such that $Pref_r(O)\subset Pref_r(L)$ then,
  Consensus can not be solved in $L$ in less than $r$ steps.
\end{proposition}
Indeed if every prefixes of length $r$ of $O$ in which Consensus is
unsolvable are also prefixes of $L$, then after $r$ rounds of any scenario
$w\in L$, processes solving Consensus in $L$ would be mean it is also solvable in $O$.

Now, using Proposition~\ref{complexity} for the fourth and fifth cases of Example~\ref{ex7cases} yields the following summary
by remarking that their first round are exactly the same as $\Gamma^\omega$.

\begin{enumerate}
\item $S_0$ is solvable in $1$ round,
\item $T_\blanc$ is solvable in $1$ round,
\item $T_\noir$ is solvable in $1$ round,
\item $C_1$ is solvable in exactly $2$ rounds,
\item $S_1$ is solvable in exactly $2$ rounds.
\end{enumerate}

\subsection{About Minimal Obstructions}

Theorem~\ref{thm:FG11} shows that, even in the simpler subclass where
no double omission are permitted, simple inclusion-minimal adversaries may not
exist. Indeed, there exists a sequence of unfair scenarios
$(u_i)_{i\in\N)}$ such that $\forall i,j$, $(u_i,u_j)$ is not a
special pair. Therefore $L_n=\Gamma^\omega \backslash
{\bigcup}_{0\leq i\leq n}u_i$ defines an infinite
decreasing sequence of obstructions for the Coordinated Attack
Problem.

Considering the set of words in $SPair(\Gamma^\omega)$,
it is possible by picking up only one member of a special pair
to have an infinite set $U$ of
unfair scenarios such that,  by
Theorem~\ref{thm:FG11}, the adversary
$\Gamma^\omega\backslash U$ is a minimal obstruction.
So there is no minimum obstruction.

As a partial conclusion, we shall say that the well known adversary
$\Gamma^\omega$, even not being formally a minimal obstruction, could
be considered, without this being formally defined, as the smallest
example of a \emph{simple obstruction}, as it is more straightforward to describe than, say the adversaries
$\Gamma^\omega\backslash U$ above. 

This probably explains why
the other obstructions we present here have never been investigated
before.

\section{Combinatorial Characterization of Solvable Adversaries}

In this part, we consider the adversaries without double omission, that is
the adversaries $L\subset\Gamma^\omega$ and we characterize exactly which one
are solvable. When the consensus is achievable, we
also give an effective algorithm. In this section, we rely on a classical combinatorial bivalency technique \cite{FLP85}.

\subsection{Necessary Condition: a Bivalency Impossibility Proof}
\label{CN}

We will use a standard, self-contained, bivalency proof technique.
We suppose now on that there is an algorithm \algo to solve Consensus
on $L$.

We proceed by contradiction. So we suppose that all the conditions of
Theorem~\ref{thm:FG11} are not true, \ie that
$Fair(\Gamma^\omega)\cup
\{\lnoir^\omega,\lblanc^\omega\} \subset L,$
and that for all $(w,w')\in SPair(\Gamma^\omega), w\notin L
\Longrightarrow w'\in L.$

\begin{definition}
  Given an initial configuration, let $v\in Pref(L)$ and
  $i\in\{0,1\}$. The partial scenario $v$ is
  said to be \emph{$i$-valent} if, for all scenario $w\in L$ such that $v\in
  Pref(w)$, \algo decides $i$ at the end. If $v$ is not $0-$valent nor
  $1-$valent, then it is said \emph{bivalent}.
\end{definition}

By hypothesis, $\lblanc^\omega,\lnoir^\omega\in L$. Consider
our algorithm \algo with input 0 on both processes running under
scenario $\lblanc^\omega$. The algorithm terminates and has to output
$0$ by Validity Property. Similarly, both processes output $1$ under
scenario $\lnoir^\omega$ with input $1$ on both processes.

From now on, we have this initial configuration $I$: $0$ on process \blanc
and $1$ on process \noir. For \blanc there is no difference under
scenario $\lnoir^\omega$ for this initial configurations and the
previous one. Hence,
\algo will output $0$ for scenario $\lnoir^\omega$. Similarly, for
\noir there is no difference under scenario $\lblanc^\omega$ for both
considered initial
configurations. Hence, \algo will output $1$ for this other scenario.
Hence $\varepsilon$ is bivalent for initial configuration $I$. In the
following, valency will always be implicitly defined with respect to
the initial configuration $I$.

\begin{definition}
  Let $v\in Pref(L)$, $v$ is \emph{decisive} if
  \begin{theoenum}
  \item $v$ is bivalent,
  \item For any $a\in\Gamma$ such that $va\in Pref(L)$, $va$ is not
    bivalent.
  \end{theoenum}
\end{definition}

\begin{lemma}
  There exist a decisive $v\in Pref(L).$
\end{lemma}
\begin{proof}
  We suppose this not true. Then, as $\varepsilon$ is bivalent, it is
  possible to construct $w\in\Gamma^\omega$ such that, $Pref(w)\subset
  Pref(L)$ and for any $v\in Pref(w)$, $v$ is bivalent. Bivalency for
  a given $v$ means in particular that the algorithm \algo has not
  stopped yet. Therefore $w\notin L$.

  This means that $w$ is unfair, because from initial assumption, $Fair(\Gamma^\omega)\subset L$. Then this means that,
  w.l.o.g we have $u\in\Gamma^+$, such that $w=u\lblanc^\omega$ and
  $ind(u)$ is even. We denote by
  $w'=ind^{-1}(ind(u)-1)\lblanc^\omega$.
  The couple $(w,w')$ is a special pair. Therefore $w'$ belongs to
  $L$, so \algo halts at some round $r_0$ under scenario $w'$.

  By Corollary~\ref{state}, $s_\noir(w_{|r_0})=s_\noir(w'_{|r_0})$
  This means that $w|_{r_0}$ is not bivalent. As $w|_{r_0}\in Pref(L)$
  this gives a contradiction.
\end{proof}

We can now end the impossibility proof.

\begin{proof}
  Consider a decisive $v\in Pref(L)$. By definition, this means that
  there exist $a,b\in\Gamma$, with $a\neq b$, such that $va,vb\in
  Pref(L)$ and $va$ and $vb$ are not bivalent and are of different
  valencies. Obviously there is an extension that has a different valency of the one of $\lok$
  and w.l.o.g., we choose $b$ to be $\lok$. Therefore $a=\lblanc$ or $a=\lnoir$.
  We terminate by a case by case analysis.

  Suppose that $a=\lblanc,$ and $ind(v)$ is even.
  By Corollary~\ref{state}, this means \noir is in the same state
  after $va$ and $vb$. We consider the scenarios $va\lblanc^\omega$
  and $vb\lblanc^\omega$, they forms a special pair. So one belongs to $L$
  by hypothesis, therefore both processes should halt at some point under
  this scenario. However, the state of \noir is always the same
  under the two scenarios because it receives no message at all
  so if it halts and decides some value, we get
  a contradiction with the different valencies.

  Other cases are treated similarly using the other cases of Corollary~\ref{state}.
\end{proof}

\subsection{A Consensus Algorithm}
\label{CS}
Given a word $w$ in $\Gamma^\omega$, we define the following algorithm
$\algo_w$ (see Algorithm~\ref{consalgo}).
It has messages always of the same type. They
have two components, the first one is the initial bit, named
$init$. The second is an integer named $ind$. Given a message $msg$,
we note $msg.init$ (resp. $msg.ind$) the first (resp. the second)
component of the message.

\begin{algorithm}
  \KwData{$w\in\Gamma^\omega$}
  \KwIn{$init\in\{0,1\}$}

  r=0\;
  initother=null\;
  \eIf{$\proc=\blanc$}{ind=0\;}{ind=1\;}
  \While{$|ind - ind(w|_r)| \leq 2$}{
    msg = (init,ind)\;
    send(msg)\;
    msg = receive()\;
    \eIf(// message was lost){msg == null}
    {$ind = 3*ind$\;}
    {$ind = 2*msg.ind+ind$\;
    initother = msg.init\;}
    r=r+1\;
  }

  \eIf{$\proc=\blanc$}{%
    \eIf{$ind < ind(w|_r)$}
    {\KwOut{init}}
    {\KwOut{initother}}
  }{
    \eIf{$ind > ind(w|_r)$}
    {\KwOut{init}}
    {\KwOut{initother}}
  }
    \caption{\label{consalgo}Consensus Algorithm $\algo_w$ for Process \proc:}
\end{algorithm}

We prove that the $ind$ values
computed by each process in the algorithm differ by only one.
Moreover, we show that the actual index is equal to the minimum of the $ind$ values.

More precisely, with $sign(n)$ being $+1$ (resp. $-1$) when $n\in\Z$ is
positive (resp. negative), we have 
\begin{proposition}\label{indexes}
  For any round $r$ of an execution of Algorithm $\algo_w$ under
  scenario $v\in\Gamma^{r}$, such that no process has already halted,
  $$\left\{\mbox{
      \begin{minipage}[c]{0.6\linewidth}
        $|ind^\noir_r-ind^\blanc_r| = 1$,\\
        $sign(ind^\noir_r-ind^\blanc_r) = (-1)^{ind(v)}$,\\
        $ind(v) = \min\{ind^\blanc_r,ind^\noir_r\}.$
      \end{minipage}}\right.
    $$
\end{proposition}
\begin{proof}
  We prove the result by induction over $r\in\N$.

  For $r=0$, the equations are satisfied.

  Suppose the property is true for $r-1$. We consider a round of the
  algorithm. Let $u\in\Gamma^{r-1}$ and $a\in\Gamma$, and consider
  an execution under environment $w=ua$. There are
  exactly three cases to consider.

  Suppose $a=\lok$. Then it means both messages are received and
  $ind^\blanc_r=2ind^\noir_{r-1}+ind^\blanc_{r-1}$ and
  $ind^\noir_r=2ind^\blanc_{r-1}+ind^\noir_{r-1}.$ Hence
  $ind^\noir_r-ind^\blanc_r = ind^\blanc_{r-1}-ind^\noir_{r-1}$.
  Thus, using the recurrence property, we get
  $|ind^\noir_r-ind^\blanc_r| = 1.$

  Moreover, by construction
  \begin{eqnarray*}
    ind(v)&=&ind(ua)\\
    &=&3ind(u)+(-1)^{ind(u)}(\mu(a))+1\\
    &=&3ind(u)+1
  \end{eqnarray*}

  The two indices $ind(u)$ and $ind(v)$ are therefore of
  opposite parity. Hence, by induction property,
  $sign(ind^\noir_r-ind^\blanc_r) = (-1)^{ind(v)}$.
  And remarking that $\min\{ind^\blanc_r,ind^\noir_r\} =$
  \begin{eqnarray*}
    &&
    \min\{2ind^\noir_{r-1}+ind^\blanc_{r-1},2ind^\blanc_{r-1}+ind^\noir_{r-1}\}\\
    &=&ind^\noir_{r-1}+ind^\blanc_{r-1}+\min\{ind^\blanc_{r-1},ind^\noir_{r-1}\}\\
    &=&2\min\{ind^\blanc_{r-1},ind^\noir_{r-1}\}+|ind^\noir_{r-1}-ind^\blanc_{r-1}|\\
    & &           +\min\{ind^\blanc_{r-1},ind^\noir_{r-1}\}\\
    &=&3\min\{ind^\blanc_{r-1},ind^\noir_{r-1}\}+1
  \end{eqnarray*}
  we get that the third equality is also verified in round $r$.

   Consider now the case $a=\lnoir$. Then \blanc gets no message from
  \noir and \noir gets a message from \blanc. So we have that
  $ind^\blanc_r=3ind^\blanc_{r-1}$ and
  $ind^\noir_r=2ind^\blanc_{r-1}+ind^\noir_{r-1}.$
  So we have that $ind^\noir_r-ind^\blanc_r = ind^\noir_{r-1}-ind^\blanc_{r-1}$.
  The first equality is satisfied.

  We also have that $ind(v)=3ind(u)+(-1)^{ind(u)}\mu(\lnoir)+1=
  3ind(u)+\alpha,$ with $\alpha=0$ if $ind(u)$ is even and $\alpha=2$
  otherwise. Hence, $ind(u)$ and $ind(v)$ are of the same parity, and
  the second equality is also satisfied.

  Finally, we have that $\min\{ind^\blanc_r,ind^\noir_r\}=
  \{3ind^\blanc_{r-1},2ind^\blanc_{r-1}+ind^\noir_{r-1}\}$. If
  $ind(u)$ is even, then $ind(u)=ind^\blanc_{r-1}$ and the equality
  holds. If $ind(u)$ is odd, then $ind(u)=ind^\noir_{r-1}$ and the
  equality also holds.

  The case $a=\lblanc$ is a symmetric case and is proved similarly.
\end{proof}

\subsection{Correctness of the Algorithm}

Given a message adversary $L$, we suppose that one of the following holds.
  \begin{itemize}
  \item $\exists f\in Fair(\Gamma^\omega), f\notin L,$
  \item $\exists (u,u')\in SPair(\Gamma^\omega), u,u'\notin L,$
  \item $\lnoir^\omega\notin L,$
  \item $\lblanc^\omega\notin L.$
  \end{itemize}

In particular, $L\subsetneq\Gamma^\omega$ and we denote by $w$, a scenario
in $\Gamma^\omega\backslash L$. If we can choose $w$ such that it is fair, we choose such a $w$.
Otherwise $w$ is unfair, and we assume it is
either $\lnoir^\omega$ or $\lblanc^\omega$, or it belongs to a special
pair that is not included in $L$.

We consider the algorithm $\algo_w$ with parameter $w$
as defined above.

\begin{lemma}\label{diff2}
  Let $v\in L$. There exist $r\in\N$ such that $|ind(v|_r)-ind(w|_r)|\geq3.$
\end{lemma}
\begin{proof}
  Given $w\notin L$, we have $v\neq w$ and at some round $r$, $w|_r\neq v|_r.$
  Therefore $|ind(v|_r)-ind(w|_r)|\geq1$.

  From Lemma~\ref{diff1} and Definition~\ref{pair}, it can be seen
  that the only way to remain indefinitely at a difference of one is
  exactly that $w$ and $v$ form a special pair. Given the way we
  have chosen $w$, and that $v\in L$, this is impossible.
  So at some round $r'$, the difference will be greater than 2 :
  $|ind(v|_{r'})-ind(w|_{r'})|\geq2$.

  Then by definition of the index we have that
  $|ind(v|_{r'+1})-ind(w|_{r'+1})|\geq3$.
\end{proof}

Now we prove the correctness of the algorithm under the message adversary $L$.

\begin{proposition}
  The algorithm $\algo_w$ is correct for every $v\in L$.
\end{proposition}

\begin{proof}
  First we show Termination. This is a corollary of
  Lemma~\ref{diff2} and Proposition~\ref{indexes}.

  Consider the execution $v\in L$.
  From Lemma~\ref{diff2}, there exists a round $r\in\N$ such that $|ind(v|_r)-ind(w|_r)|\geq3.$

  Denote by $r$ the round when it is first
  satisfied for one of the process. From the condition of the While loop, it means this process will stop at round $r$. 
  If the $ind$ value of the other process \proc at round $r$ is also at distance 3 or more from the index of $w|_r$,
  then we are done. Otherwise, 
  from Proposition~\ref{indexes}, we have
  $|ind^\proc_r-ind(w|_r)| = 2.$
  In the following round, \proc
  will receive no message (the other process has halted) and
  $|ind^\proc_r-ind(w|_r)| \geq3$ will hold .
  Note also, that even though the final comparisons are strict, the output value is well defined at the end
  since $ind_r\neq ind(w|_r)$ at this stage.

  The validity property is also obvious, as the output values are ones
  of the initial values.  Consider the case for \noir (the case for \blanc is symmetric).
  Since the only case where
  $initother^\noir$ would be $null$ is when \noir has
  not received any message from \blanc, \ie when
  $v=\lblanc^r$. But as $ind(\lblanc^r)=3^r-1$ is the maximal possible
  index for scenario of length $r$, and \noir outputs  $initother^\noir$ only if
  $ind^\noir < ind(w|_r)$, it cannot have to output
  $initother$ when it is $null$.
  Similarly for \blanc, this proves that $null$ can never be output by any process.

  We now prove the agreement property. Given that
  $|ind^\blanc_r-ind^\noir_r|=1$ by Proposition~\ref{indexes}, when
  the processes halt, from Lemma~\ref{diff2} the $ind$ values are on the same side of
  $ind(w{|_r})$. This means that one of the process outputs $init$, the
  other outputting $initother$. By construction, they
  output the same value.
\end{proof}

\section{Topological approach}
\label{sec:topo}
In this Section, we provide a topological characterization of solvable message
adversaries for the Consensus Problem.
First, we will introduce some basic topological
definitions in Section~\ref{sub:def_topo}, then we will explain the link
between topology and distributed computability in
Section~\ref{sub:topo_rep_of_dist_syst} in order to formulate our result in
Section~\ref{sub:topo_carac}.

We then show in the following
Section~\ref{sub:link_with_the_combinatorial_characterization} how this new
characterization matches the combinatorial one described by
Theorem~\ref{thm:FG11}.
We also discuss a similar characterization in \cite{GKM14} and we show
that our result indicates that there is a flaw in the statement of
Theorem 6.1 of \cite{GKM14}. If no restriction are given on the kind of adversary are addressed in this theorem,
then, in the case of 2 processes, this would imply that $\Gamma^\omega\backslash\{w\}$ is solvable for any $w$.
From our result, this is incorrect when $w$ belongs to a special pair.
This has been confirmed by the authors \cite{K15} that the statement has to be corrected by restricting to adversaries $L$ that are closed for special pairs (if $\{w,w'\}$ is a special pair, then $w\in L \Leftrightarrow w'\in L$).
Moreover, even if the present work is for only 2 processes,
the approach that is taken might help correct the general statement of \cite{GKM14}.

\subsection{Definitions}
\label{sub:def_topo}

The following definitions
are standard definitions from algebraic topology \cite{Munkres84}.
We fix an integer $N\in\N$ for this part. 

\begin{definition}
  Let $n\in\N$.
  A finite set $\sigma=\{v_0,\dots,v_n\}\subset\R^N$ is called a \emph{simplex} of dimension $n$ if the vector space generated by
  $\{v_1-v_0,\dots,v_n-v_0\}$ is of dimension $n$.
  We denote by $|\sigma|$ the convex hull of $\sigma$ that we call
  the \emph{geometric realization} of $\sigma$.
\end{definition}

\begin{definition} 
A \emph{simplicial complex} is a collection $C$ of \emph{simplices}
such that~:
\begin{enumerate}[(a)]
  \item If $\sigma\in C$ and $\sigma'\subseteq\sigma$, then $\sigma'\in C$,
  \item If $\sigma,\tau\in C$ and $|\sigma|\cap|\tau|\neq\emptyset$ then there exists $\sigma'\in C$ such that
    \begin{itemize}
    \item $|\sigma|\cap|\tau|=|\sigma'|$,
    \item $\sigma'\subset\sigma, \sigma'\subset\tau.$
    \end{itemize}
\end{enumerate}
\end{definition}

The simplices of dimension 0 (singleton) of $C$ are called vertices,
we denote $V(C)$ the set of vertices.
The \emph{geometric realization} of $C$, denoted $|C|$, is the union
of the geometric realization of the simplices of $C$.

Let $A$ and $B$ be simplicial complexes. A map $f\colon V(A)\to V(B)$ is called
\emph{simplicial} (in which case we write $f\colon A\to B$) if it preserves the
simplices, \ie for each simplex $\sigma$ of $A$, the image $f(\sigma)$ is a
simplex of $B$.

In this paper, we also work with colored simplicial complexes. These
are simplicial complexes $C$ together with 
a function $c:V(C)\to \Pi$ such that the restriction of $c$ on any
maximal simplex of $C$ is a bijection. A simplicial map that preserves colors is called chromatic.

As a final note, since we only deal with two processes, our simplicial complexes will be of
dimension 1. The only simplices are edges (sometimes called \emph{segments})
and vertices, and the latter are colored with $\Pi=\{\b,\n\}$.

\begin{remark}\label{subtle}  
  The combinatorial part of simplicial complexes (that is the sets of
  vertices and the inclusion relationships they have) is usually
  referred as abstract simplicial complexes. Abstract simplicial complex can be equivalently defined as a
  collection $C$ of sets that are closed by inclusion, that is if
  $S\in C$ and if $S'\subset S$ then $S'\in C$.

  For finite complexes, the
  topological and combinatorial notions are equivalent, including
  regarding geometric realizations. 
  But it should be noted that infinite abstract simplicial complex might not
  have a unique (even up to homeomorphisms) geometric realization.
  However, here the infinite complexes we deal with
  are derived from the
  subdivision, see below, of a finite initial complex, hence they have a unique
  realization in that setting. So in the context of this paper, we
  will talk about ``the'' realization of such infinite complexes.
\end{remark}

The elements of $|C|$ can be expressed as convex combinations of the
vertices of $C$, \ie $\forall x\in|C|\quad x=\underset{v\in V(C)}{\sum}\alpha_v
v$ such that $\underset{v\in V(C)}{\sum}\alpha_v=1$ and
$\{v\mid \alpha_v\neq0\}$ is a simplex of $C$.

The geometric realization of a simplicial map $\delta:A\to B$, is
$|\delta|:|A|\to|B|$ and is obtained by 
$f(x)=\underset{v\in V(C)}{\sum}\alpha_vf(v)$.

A \emph{subdivision} of a simplicial complex $C$ is a simplicial
complex $C'$ such that~:
\begin{enumerate}[(1)]
  \item the vertices of $C'$ are points of $|C|$,
  \item for any $\sigma'\in C'$, there exists $\sigma\in C$ such that
    $|\sigma'|\subset|\sigma|$;
  \item $|C|=|C'|$.
\end{enumerate}

Let $C$ be a chromatic complex of dimension 1, its \emph{standard chromatic subdivision}
$\text{Chr~}C$ is obtained by replacing each simplex $\sigma\in C$ by its
chromatic subdivision. See \cite{HKRbook} for the general definition
of the chromatic subdivision of a simplex, we present here only
the chromatic subdivision of a segment $[0,1]\subset\R$ whose vertices are colored $\b$
and $\n$ respectively.
The subdivision is defined 
as the chromatic complex consisting of 4 vertices at position
$0$, $\frac{1}{3}$, $\frac{2}{3}$ and $1$; and colored $\b$, $\n$, $\b$, $\n$ respectively. The edges are the 3 segments
$[0,\frac{1}{3}]$, $[\frac{1}{3}, \frac{2}{3}]$ and $[\frac{2}{3},1]$.
The geometric realization of the chromatic subdivision of the segment $[0,1]$
is identical to the segment's.

If we iterate this
process $m$ times we obtain the $m$\textsuperscript{th} chromatic subdivision
denoted $\text{Chr}^m~C$.
Figure~\ref{fig:sub_segment}
shows $\text{Chr~}[0,1]$ to $\text{Chr}^3~[0,1]$.

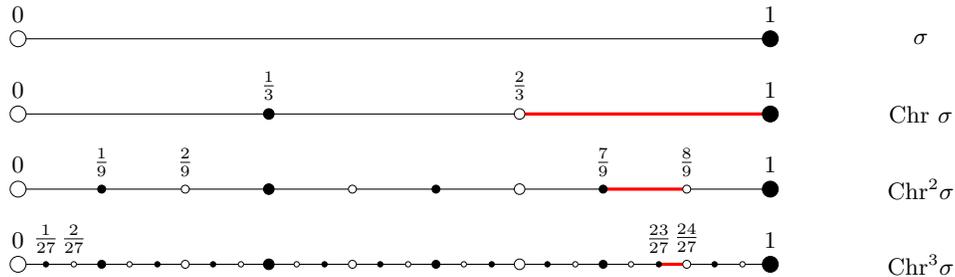
\begin{figure}[ht]
  \centering
  \begin{tikzpicture}[]

\node at (12, 3) {$\sigma$};
\draw (0, 3) -- (10, 3);
\node[proc,b,minimum width=6pt,label=0] at (0, 3) {};
\node[proc,n,minimum width=6pt,label=1] at (10, 3) {};

\node at (12, 2) {$\text{Chr~}\sigma$};
\draw[shorten <=3pt,shorten >=2pt] (0, 2) -- (3.3333333333333335, 2.0);
\draw[shorten <=2pt,shorten >=2pt] (3.3333333333333335, 2.0) -- (6.666666666666666, 2.0);
\draw[very thick,red,shorten <=2pt,shorten >=3pt] (6.666666666666666, 2.0) -- (10, 2);
\node[proc,b,minimum width=6pt,label=0] at (0, 2) {};
\node[proc,n,minimum width=6pt,label=1] at (10, 2) {};
\node[proc,b,minimum width=4pt,label=$\frac{2}{3}$] at (6.666666666666666, 2.0) {};
\node[proc,n,minimum width=4pt,label=$\frac{1}{3}$] at (3.3333333333333335, 2.0) {};

\node at (12, 1) {$\text{Chr}^2\sigma$};
\draw[shorten <=3pt,shorten >=1.5pt] (0, 1) -- (1.1111111111111112, 1.0);
\draw[shorten <=1.5pt,shorten >=1.5pt] (1.1111111111111112, 1.0) -- (2.2222222222222223, 1.0);
\draw[shorten <=1.5pt,shorten >=2pt] (2.2222222222222223, 1.0) -- (3.3333333333333335, 1.0);
\draw[shorten <=2pt,shorten >=1.5pt] (3.3333333333333335, 1.0) -- (4.444444444444445, 1.0);
\draw[shorten <=1.5pt,shorten >=1.5pt] (4.444444444444445, 1.0) -- (5.555555555555555, 1.0);
\draw[shorten <=1.5pt,shorten >=2pt] (5.555555555555555, 1.0) -- (6.666666666666666, 1.0);
\draw[shorten <=2pt,shorten >=1.5pt] (6.666666666666666, 1.0) -- (7.777777777777777, 1.0);
\draw[very thick,red,shorten <=1.5pt,shorten >=1.5pt] (7.777777777777777, 1.0) -- (8.88888888888889, 1.0);
\draw[shorten <=1.5pt,shorten >=3pt] (8.88888888888889, 1.0) -- (10, 1);
\node[proc,b,minimum width=6pt,label=0] at (0, 1) {};
\node[proc,n,minimum width=6pt,label=1] at (10, 1) {};
\node[proc,b,minimum width=4pt] at (6.666666666666666, 1.0) {};
\node[proc,n,minimum width=4pt] at (3.3333333333333335, 1.0) {};
\node[proc,b,minimum width=3pt,label=$\frac{2}{9}$] at (2.2222222222222223, 1.0) {};
\node[proc,n,minimum width=3pt,label=$\frac{1}{9}$] at (1.1111111111111112, 1.0) {};
\node[proc,n,minimum width=3pt] at (5.555555555555555, 1.0) {};
\node[proc,b,minimum width=3pt] at (4.444444444444445, 1.0) {};
\node[proc,b,minimum width=3pt,label=$\frac{8}{9}$] at (8.88888888888889, 1.0) {};
\node[proc,n,minimum width=3pt,label=$\frac{7}{9}$] at (7.777777777777777, 1.0) {};

\node at (12, 0) {$\text{Chr}^3\sigma$};
\draw[shorten <=3pt,shorten >=1pt] (0, 0) -- (0.37037037037037046, 0.0);
\draw[shorten <=1pt,shorten >=1pt] (0.37037037037037046, 0.0) -- (0.7407407407407408, 0.0);
\draw[shorten <=1pt,shorten >=1.5pt] (0.7407407407407408, 0.0) -- (1.1111111111111112, 0.0);
\draw[shorten <=1.5pt,shorten >=1pt] (1.1111111111111112, 0.0) -- (1.4814814814814816, 0.0);
\draw[shorten <=1pt,shorten >=1pt] (1.4814814814814816, 0.0) -- (1.8518518518518519, 0.0);
\draw[shorten <=1pt,shorten >=1.5pt] (1.8518518518518519, 0.0) -- (2.2222222222222223, 0.0);
\draw[shorten <=1.5pt,shorten >=1pt] (2.2222222222222223, 0.0) -- (2.5925925925925926, 0.0);
\draw[shorten <=1pt,shorten >=1pt] (2.5925925925925926, 0.0) -- (2.962962962962963, 0.0);
\draw[shorten <=1pt,shorten >=2pt] (2.962962962962963, 0.0) -- (3.3333333333333335, 0.0);
\draw[shorten <=2pt,shorten >=1pt] (3.3333333333333335, 0.0) -- (3.7037037037037037, 0.0);
\draw[shorten <=1pt,shorten >=1pt] (3.7037037037037037, 0.0) -- (4.074074074074074, 0.0);
\draw[shorten <=1pt,shorten >=1.5pt] (4.074074074074074, 0.0) -- (4.444444444444445, 0.0);
\draw[shorten <=1.5pt,shorten >=1pt] (4.444444444444445, 0.0) -- (4.814814814814815, 0.0);
\draw[shorten <=1pt,shorten >=1pt] (4.814814814814815, 0.0) -- (5.185185185185185, 0.0);
\draw[shorten <=1pt,shorten >=1.5pt] (5.185185185185185, 0.0) -- (5.555555555555555, 0.0);
\draw[shorten <=1.5pt,shorten >=1pt] (5.555555555555555, 0.0) -- (5.925925925925926, 0.0);
\draw[shorten <=1pt,shorten >=1pt] (5.925925925925926, 0.0) -- (6.296296296296296, 0.0);
\draw[shorten <=1pt,shorten >=2pt] (6.296296296296296, 0.0) -- (6.666666666666666, 0.0);
\draw[shorten <=2pt,shorten >=1pt] (6.666666666666666, 0.0) -- (7.037037037037036, 0.0);
\draw[shorten <=1pt,shorten >=1pt] (7.037037037037036, 0.0) -- (7.407407407407407, 0.0);
\draw[shorten <=1pt,shorten >=1.5pt] (7.407407407407407, 0.0) -- (7.777777777777777, 0.0);
\draw[shorten <=1.5pt,shorten >=1pt] (7.777777777777777, 0.0) -- (8.148148148148149, 0.0);
\draw[shorten <=1pt,shorten >=1pt] (8.148148148148149, 0.0) -- (8.518518518518517, 0.0);
\draw[very thick,red,shorten <=1pt,shorten >=1.5pt] (8.518518518518517, 0.0) -- (8.88888888888889, 0.0);
\draw[shorten <=1.5pt,shorten >=1pt] (8.88888888888889, 0.0) -- (9.259259259259258, 0.0);
\draw[shorten <=1pt,shorten >=1pt] (9.259259259259258, 0.0) -- (9.62962962962963, 0.0);
\draw[shorten <=1pt,shorten >=3pt] (9.62962962962963, 0.0) -- (10, 0);
\node[proc,b,minimum width=6pt,label=0] at (0, 0) {};
\node[proc,n,minimum width=6pt,label=1] at (10, 0) {};
\node[proc,b,minimum width=4pt] at (6.666666666666666, 0.0) {};
\node[proc,n,minimum width=4pt] at (3.3333333333333335, 0.0) {};
\node[proc,b,minimum width=3pt] at (2.2222222222222223, 0.0) {};
\node[proc,n,minimum width=3pt] at (1.1111111111111112, 0.0) {};
\node[proc,b,minimum width=2pt,label=$\frac{2}{27}$] at (0.7407407407407408, 0.0) {};
\node[proc,n,minimum width=2pt,label=$\frac{1}{27}$] at (0.37037037037037046, 0.0) {};
\node[proc,n,minimum width=2pt] at (1.8518518518518519, 0.0) {};
\node[proc,b,minimum width=2pt] at (1.4814814814814816, 0.0) {};
\node[proc,b,minimum width=2pt] at (2.962962962962963, 0.0) {};
\node[proc,n,minimum width=2pt] at (2.5925925925925926, 0.0) {};
\node[proc,n,minimum width=3pt] at (5.555555555555555, 0.0) {};
\node[proc,b,minimum width=3pt] at (4.444444444444445, 0.0) {};
\node[proc,n,minimum width=2pt] at (4.074074074074074, 0.0) {};
\node[proc,b,minimum width=2pt] at (3.7037037037037037, 0.0) {};
\node[proc,b,minimum width=2pt] at (5.185185185185185, 0.0) {};
\node[proc,n,minimum width=2pt] at (4.814814814814815, 0.0) {};
\node[proc,n,minimum width=2pt] at (6.296296296296296, 0.0) {};
\node[proc,b,minimum width=2pt] at (5.925925925925926, 0.0) {};
\node[proc,b,minimum width=3pt,label=$\frac{24}{27}$] at (8.88888888888889, 0.0) {};
\node[proc,n,minimum width=3pt] at (7.777777777777777, 0.0) {};
\node[proc,b,minimum width=2pt] at (7.407407407407407, 0.0) {};
\node[proc,n,minimum width=2pt] at (7.037037037037036, 0.0) {};
\node[proc,n,minimum width=2pt,label=$\frac{23}{27}$] at (8.518518518518517, 0.0) {};
\node[proc,b,minimum width=2pt] at (8.148148148148149, 0.0) {};
\node[proc,b,minimum width=2pt] at (9.62962962962963, 0.0) {};
\node[proc,n,minimum width=2pt] at (9.259259259259258, 0.0) {};

\end{tikzpicture}
  \caption{Chromatic subdivision of the segment.\\
    The correspondence with executions (bolder simplices) is explained at the end of section
  \ref{sub:topo_rep_of_dist_syst}.
  \label{fig:sub_segment}}
\end{figure}

\subsection{Topological representation of distributed systems}
\label{sub:topo_rep_of_dist_syst}

As shown from the celebrated Asynchronous Computability Theorem
\cite{HS99}, it is possible to encode the global state of a
distributed system in simplicial complexes. First we give the
intuition for the corresponding abstract simplicial complex.

We can represent a process in a given state by a vertex composed of a color and a value~:
the \emph{identity} and the \emph{local state}.
A \emph{global configuration} is an edge whose
vertices are colored $\b$ and $\n$.

The set of input vectors of a distributed task can thus be represented by a
colored simplicial complex $\I$ (associating all possible global initial configurations).

\begin{remark}
The vertices that belongs to more than one edge illustrate
the uncertainty of a process about the global state of a distributed system. In other words, a local state can be common
to multiple global configurations, and the process does not know in which
configuration it belongs. We have a topological (and geometrical)
representation of theses uncertainties.

For example, Figure~\ref{fig:ex_incertitude} shows a very simple graph where
each colored vertex is associated with a value (0 or 1). The vertex in the
middle is common to both edges; it represents the uncertainty of the process
$\n$ concerning the value of $\b$, \ie $\n$ doesn't know if it is in the global
configuration $(0,0)$ of $(0,1)$

\begin{figure}[ht]
  \centering
  \begin{tikzpicture}[thick]
  \node[proc,b] (b0) at (0,0)    [label=0] {};
  \node[proc,n] (n0) at (2,0)    [label=0] {};
  \node[proc,b] (b1) at (4,0)    [label=1] {};
  \draw (b0) -- (n0) -- (b1);
  % \node[draw=none,below] {0};
\end{tikzpicture}
  \caption{Example of a simplicial complex with uncertainty}
  \label{fig:ex_incertitude}
\end{figure}
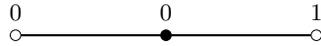

\end{remark}

In the same way than $\I$, we construct (for a distributed task) the output
complex $\O$ that contains all possible output configurations of the processes.
For a given problem, it is possible to construct a relation $\Delta\subset\I\times\O$ that is chromatic and relates the input
edges with the associated possible output edges. So, any task can
be topologically represented by a triplet $(\I,\O,\Delta)$.

\begin{figure}[ht]
  \centering
  \begin{tikzpicture}[>=stealth]

  \node[proc,n] (In0) at (0,0)  [] {};
  \node[proc,b] (Ib0) at (4,0)  [] {};
  \node[proc,b] (Ib1) at (0,4)  [label=1] {};
  \node[proc,n] (In1) at (4,4)  [label=1] {};
  \draw (In0) -- (Ib0) -- (In1) -- (Ib1) -- (In0);
  \node[draw=none,below=1pt of In0] {0};
  \node[draw=none,below=1pt of Ib0] {0};
  \draw (2,-1) node{$\I_{2gen}$};

  \node[proc,n] (On0) at (10,0)  [] {};
  \node[proc,b] (Ob0) at (14,0)  [] {};
  \node[proc,b] (Ob1) at (10,4)  [label=1] {};
  \node[proc,n] (On1) at (14,4)  [label=1] {};
  \draw (On0) -- (Ob0);
  \draw (Ob1) -- (On1);
  \node[draw=none,below=1pt of On0] {0};
  \node[draw=none,below=1pt of Ob0] {0};
  \draw (12,-1) node{$\O_{2gen}$};

  \draw[->,gray,very thick] (2,0) to[bend right] (12,0);
  \draw[->,gray,very thick] (2,4) to[bend  left] (12,4);
  \draw[  ,gray,very thick] (4,3) to (8,3);
  \draw[->,gray,very thick] (8,3) to[bend right=20] (11,4);
  \draw[->,gray,very thick] (8,3) to[bend  left=20] (13,0);
  \draw[  ,gray,very thick] (0,1) to (8,1);
  \draw[->,gray,very thick] (8,1) to[bend right=20] (13,4);
  \draw[->,gray,very thick] (8,1) to[bend  left=20] (11,0);
  \draw (7,-1) node[gray] {$\Delta_{2gen}$};

\end{tikzpicture}
  \caption{Representation of the Consensus Task with 2 processes}
  \label{fig:cons_bin_2proc}
\end{figure}
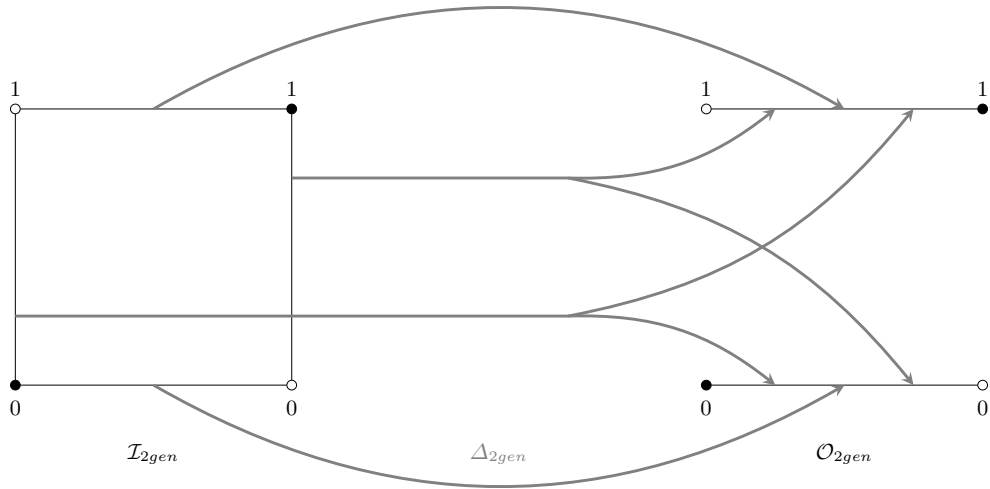

For example, the
Binary Consensus task with two processes, noted $(\I_{2gen},\O_{2gen},\Delta_{2gen})$,
is shown in Figure~\ref{fig:cons_bin_2proc}.
The input complex $\I_{2gen}$, on the left-hand side, consists of a square. Indeed, there are only four
possible global configurations given that the two process can only be in two
different initial states (proposed value 0 or 1).
The output complex, on the right-hand side, has only two edges
corresponding to the valid configurations for the Consensus (all 0 or all 1).
Finally $\Delta$ maps the input configuration with the possible output ones,
according to the \emph{validity} property of the Consensus.
\bigskip

Any protocol can be encoded as a full information protocol
\footnote{since the full information
protocol send all possible information, 
the computations can be emulated provided it is allowed to send so much information.}
any local state is related to what is called the view of the execution.
A protocol simplex is a global configuration such that there exists an
execution of the protocol in which the processes end with theses states. The
set of all theses simplices forms the \emph{protocol complex} associated to an
input complex, a set of executions and an algorithm. Given any algorithm, it
can be redefined using a full-information protocol, the protocol
complex thus depends only on the input complex and the set of executions.

Given an input complex $\I$, we construct the \emph{protocol complex at step $r$}
by applying  $r$ times to each simplex $\sigma$
the chromatic subdivision shown in
Section~\ref{sub:def_topo} and Figure~\ref{fig:sub_segment}.

\bigskip
The input complex of the Consensus and the
first two steps of the associated protocol complex are shown
in Figure~\ref{fig:cpx_proto_cons32}.

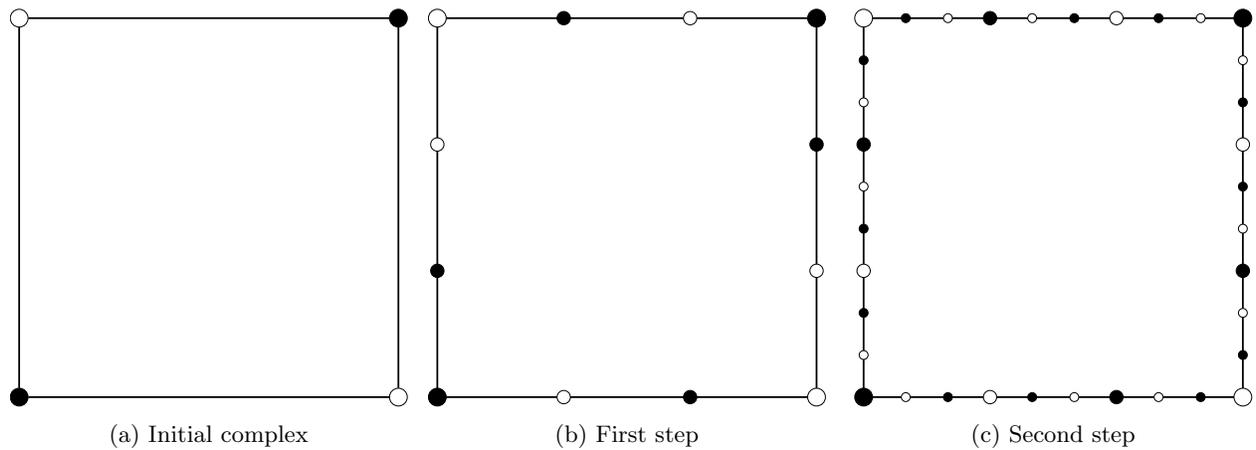
\begin{figure}[ht]
  \centering
  \begin{subfigure}{.32\textwidth}
    \begin{tikzpicture}[scale=0.3]

  \node[proc,n] (I00) at (0 ,0 )  [] {};
  \node[proc,b] (I30) at (10,0 )  [] {};
  \node[proc,n] (I33) at (10,10)  [] {};
  \node[proc,b] (I03) at (0 ,10)  [] {};

  \draw (I00) -- (I30);
  \draw (I30) -- (I33);
  \draw (I33) -- (I03);
  \draw (I03) -- (I00);

\end{tikzpicture}
    \caption{Initial complex}
  \end{subfigure}%
  \hfill
  \begin{subfigure}{.32\textwidth}
    \begin{tikzpicture}[scale=0.3]

\draw (0, 0) -- (3.3333333333333335, 0.0);
\draw (3.3333333333333335, 0.0) -- (6.666666666666666, 0.0);
\draw (6.666666666666666, 0.0) -- (10, 0);
\draw (10, 0) -- (10.0, 3.3333333333333335);
\draw (10.0, 3.3333333333333335) -- (10.0, 6.666666666666666);
\draw (10.0, 6.666666666666666) -- (10, 10);
\draw (10, 10) -- (6.666666666666666, 10.0);
\draw (6.666666666666666, 10.0) -- (3.3333333333333335, 10.0);
\draw (3.3333333333333335, 10.0) -- (0, 10);
\draw (0, 10) -- (0.0, 6.666666666666666);
\draw (0.0, 6.666666666666666) -- (0.0, 3.3333333333333335);
\draw (0.0, 3.3333333333333335) -- (0, 0);
\node[proc ,n] at (0, 0) {};
\node[proc ,b] at (10, 0) {};
\node[proc ,n] at (10, 10) {};
\node[proc ,b] at (0, 10) {};
\node[proc1,n] at (6.666666666666666, 0.0) {};
\node[proc1,b] at (3.3333333333333335, 0.0) {};
\node[proc1,n] at (10.0, 6.666666666666666) {};
\node[proc1,b] at (10.0, 3.3333333333333335) {};
\node[proc1,n] at (3.3333333333333335, 10.0) {};
\node[proc1,b] at (6.666666666666666, 10.0) {};
\node[proc1,n] at (0.0, 3.3333333333333335) {};
\node[proc1,b] at (0.0, 6.666666666666666) {};

\end{tikzpicture}
    \caption{First step}
  \end{subfigure}
  \hfill
  \begin{subfigure}{.32\textwidth}
    \begin{tikzpicture}[scale=0.3]

\draw (0, 0) -- (1.1111111111111112, 0.0);
\draw (1.1111111111111112, 0.0) -- (2.2222222222222223, 0.0);
\draw (2.2222222222222223, 0.0) -- (3.3333333333333335, 0.0);
\draw (3.3333333333333335, 0.0) -- (4.444444444444445, 0.0);
\draw (4.444444444444445, 0.0) -- (5.555555555555555, 0.0);
\draw (5.555555555555555, 0.0) -- (6.666666666666666, 0.0);
\draw (6.666666666666666, 0.0) -- (7.777777777777777, 0.0);
\draw (7.777777777777777, 0.0) -- (8.88888888888889, 0.0);
\draw (8.88888888888889, 0.0) -- (10, 0);
\draw (10, 0) -- (10.0, 1.1111111111111112);
\draw (10.0, 1.1111111111111112) -- (10.0, 2.2222222222222223);
\draw (10.0, 2.2222222222222223) -- (10.0, 3.3333333333333335);
\draw (10.0, 3.3333333333333335) -- (10.0, 4.444444444444445);
\draw (10.0, 4.444444444444445) -- (10.0, 5.555555555555555);
\draw (10.0, 5.555555555555555) -- (10.0, 6.666666666666666);
\draw (10.0, 6.666666666666666) -- (10.0, 7.777777777777777);
\draw (10.0, 7.777777777777777) -- (10.0, 8.88888888888889);
\draw (10.0, 8.88888888888889) -- (10, 10);
\draw (10, 10) -- (8.88888888888889, 10.0);
\draw (8.88888888888889, 10.0) -- (7.777777777777777, 10.0);
\draw (7.777777777777777, 10.0) -- (6.666666666666666, 10.0);
\draw (6.666666666666666, 10.0) -- (5.555555555555555, 10.0);
\draw (5.555555555555555, 10.0) -- (4.444444444444445, 10.0);
\draw (4.444444444444445, 10.0) -- (3.3333333333333335, 10.0);
\draw (3.3333333333333335, 10.0) -- (2.2222222222222223, 10.0);
\draw (2.2222222222222223, 10.0) -- (1.1111111111111112, 10.0);
\draw (1.1111111111111112, 10.0) -- (0, 10);
\draw (0, 10) -- (0.0, 8.88888888888889);
\draw (0.0, 8.88888888888889) -- (0.0, 7.777777777777777);
\draw (0.0, 7.777777777777777) -- (0.0, 6.666666666666666);
\draw (0.0, 6.666666666666666) -- (0.0, 5.555555555555555);
\draw (0.0, 5.555555555555555) -- (0.0, 4.444444444444445);
\draw (0.0, 4.444444444444445) -- (0.0, 3.3333333333333335);
\draw (0.0, 3.3333333333333335) -- (0.0, 2.2222222222222223);
\draw (0.0, 2.2222222222222223) -- (0.0, 1.1111111111111112);
\draw (0.0, 1.1111111111111112) -- (0, 0);

\node[proc ,n] at (0, 0) {};
\node[proc ,b] at (10, 0) {};
\node[proc ,n] at (10, 10) {};
\node[proc ,b] at (0, 10) {};
\node[proc1,n] at (6.666666666666666, 0.0) {};
\node[proc1,b] at (3.3333333333333335, 0.0) {};
\node[proc2,n] at (2.2222222222222223, 0.0) {};
\node[proc2,b] at (1.1111111111111112, 0.0) {};
\node[proc2,b] at (5.555555555555555, 0.0) {};
\node[proc2,n] at (4.444444444444445, 0.0) {};
\node[proc2,n] at (8.88888888888889, 0.0) {};
\node[proc2,b] at (7.777777777777777, 0.0) {};
\node[proc1,b] at (10.0, 6.666666666666666) {};
\node[proc1,n] at (10.0, 3.3333333333333335) {};
\node[proc2,b] at (10.0, 2.2222222222222223) {};
\node[proc2,n] at (10.0, 1.1111111111111112) {};
\node[proc2,n] at (10.0, 5.555555555555555) {};
\node[proc2,b] at (10.0, 4.444444444444445) {};
\node[proc2,b] at (10.0, 8.88888888888889) {};
\node[proc2,n] at (10.0, 7.777777777777777) {};
\node[proc1,n] at (3.3333333333333335, 10.0) {};
\node[proc1,b] at (6.666666666666666, 10.0) {};
\node[proc2,n] at (7.777777777777777, 10.0) {};
\node[proc2,b] at (8.88888888888889, 10.0) {};
\node[proc2,b] at (4.444444444444445, 10.0) {};
\node[proc2,n] at (5.555555555555555, 10.0) {};
\node[proc2,n] at (1.1111111111111112, 10.0) {};
\node[proc2,b] at (2.2222222222222223, 10.0) {};
\node[proc1,b] at (0.0, 3.3333333333333335) {};
\node[proc1,n] at (0.0, 6.666666666666666) {};
\node[proc2,b] at (0.0, 7.777777777777777) {};
\node[proc2,n] at (0.0, 8.88888888888889) {};
\node[proc2,n] at (0.0, 4.444444444444445) {};
\node[proc2,b] at (0.0, 5.555555555555555) {};
\node[proc2,b] at (0.0, 1.1111111111111112) {};
\node[proc2,n] at (0.0, 2.2222222222222223) {};

\end{tikzpicture}
    \caption{Second step}
  \end{subfigure}
  \caption{Protocol complex of the initial, first and second rounds}
  \label{fig:cpx_proto_cons32}
\end{figure}

For example, let $w_0=\lblanc\lok\lnoir^\omega$. $w$ corresponds to the following
execution~:
\begin{itemize}
  \item[--] first, the message from $\b$ is lost;
  \item[--] $\b$ and $\n$ receive both message;
  \item[--] the message from $\n$ is lost;
  \item[--] ... (this last round is repeated indefinitely)
\end{itemize}

This can be represented as a sequence of simplices.
Each infinite sequence of edges ($\sigma_0,\sigma_1,\ldots$) with
$\sigma_{i+1}\in \text{Chr}\sigma_i$, and $\sigma_0$ the initial configuration,
corresponds to a unique scenario and vice versa.
Consider Figure~\ref{fig:sub_segment}, at $[\frac{2}{3},1]$
the thick red simplex corresponds to $\lblanc$.
Then, at $[\frac{7}{9},\frac{8}{9}]$, this corresponds to $\lblanc\lok$.
Finally, at $[\frac{23}{27},\frac{24}{27}]$
the thick red simplex of $\text{Chr}^3~\sigma$ 
corresponds to the execution $\lblanc\lok\lnoir$.

For a given finite execution, the embedding is exactly
given by the normalized index $ind_n$.
Thus, when the initial simplex is $[0,1]$ ($\b$ initial value is 0, $\n$
is 1) for an infinite execution, the
corresponding sequence of simplices will converge to a point of $[0,1]$.
At the limit, the vertex' convergence point is given by $\overline{ind}$.
For example, $w_0$ converges to $1/9$.

Without loss of generality, in the rest of the paper, we will describe
what happens on the segment $[0,1]$ instead of the whole
complex. Indeed, the behaviour of subdivision of initial segments is
identical (through a straightforward isometry for each segment) as
this behaviour does not depend on the initial values. The "gluing" at
``corners'' is always preserved since it corresponds to state obtained
by a process when it receives no message at all.

Given $x\in [0,1], i_\blanc,i_\noir\in\{0,1\}$, we denote by
$geo(x,i_\blanc,i_\noir)$ the point in the geometrical realization
that corresponds to $x$ in the simplex
$[(\blanc,i_\blanc),(\noir,i_\noir)]$, namely $x(\blanc,i_\blanc) + (1-x)(\noir,i_\noir)$.

Given $L\subset\Gamma^\omega$, we denote $|\mathcal
C^L|=\{geo(\overline{ind}(w),i_\blanc,i_\noir)\mid \exists w\in L,
i_\blanc,i_\noir\in\{0,1\}\}$. Note that it is possible to define the
limit of the geometric realizations this way, but that there are no
sensible way to define a geometric realization directly for the corresponding abstract simplicial
complex. See Section \ref{contrex}.

\bigskip

\subsection{Topological Characterization}

The following definition is inspired by \cite{GKM14}.
\begin{definition}
  Let $C$ be a colored simplicial complex.
  A \emph{terminating subdivision} $\TS$ of $C$ is a (possibly
  infinite) sequence of colored simplicial complexes
  $(\Sigma_k)_{k\in\N}$ such that $\Sigma_0=\emptyset$, and
  for all $k\geq1$ $\Sigma_k\subset Chr^k C$, and
  $\cup_{i\leq k} \Sigma_i$ is a simplicial complex for all $k$.
\end{definition}

The intuition for this definition is as follows.  It is well known
that (non-terminated) runs of length $r$ in $\Gamma^\omega$ with
initial values encoded as $C$ are represented by a protocol complex
that is the chromatic subdivision $Chr^k C$. This definition refines
the known correspondence by looking at the actual runs for a given
algorithm. When the algorithm stops, the protocol complex should
actually be no more refined. For a given algorithm, we end up with
protocol complexes that are of a special form : where, possibly, the
level of chromatic subdivision is not the same everywhere.  We will
prove later that those resulting complexes are exactly terminating
subdivisions.  Or to say it differently, terminating subdivisions are
the form of protocol complexes for non-uniformly terminating
algorithms (that is for algorithms that do not stop at the same round
for all possible executions).

From the correspondence between words and simplexes of the chromatic
subdivision, we can see that the corresponding set of words is an
anti-chain for the prefix order relation.

An edge of $\Sigma_k$ for some $k$ is called a \emph{stable
  edge} in the subdivision $\TS$.
The union
$\cup_k\Sigma_k$ of stable edges in $\TS$ forms a
colored simplicial complex, and a stable edge can only intersect
another given stable edge at an extremity (a single vertex).

For $r\in\N$, we denote by $K_r(\TS)$ the complex $\cup_{k\leq r}\Sigma_k$.
We denote by $K(\TS)$ the complex $\cup_k\Sigma_k$; it possibly has infinitely many
simplices. Observe that the geometric realization $|K(\TS)|$ can be
identified with a subset of $|C|$.
For example, Figure~\ref{fig:ex_sub_ter} shows a terminating subdivision of
$[0,1]$ up to round $3$.

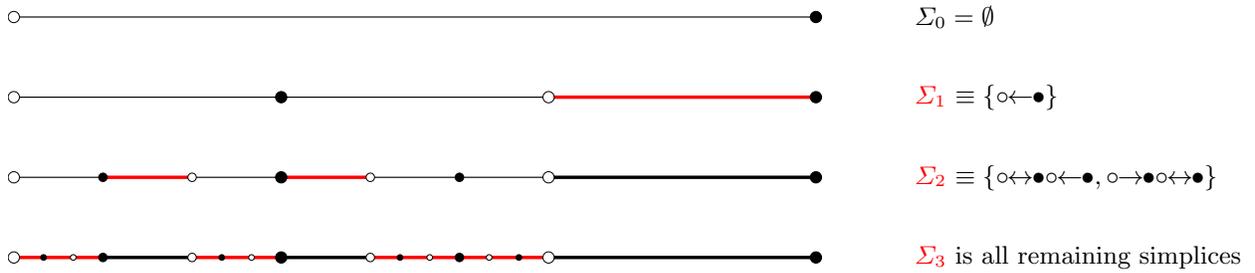
\begin{figure}[ht]
  \centering
  \begin{tikzpicture}

%\node[label={right:$C_0$}] at (11, 3) { };
\node[label={right:$\Sigma_0=\emptyset$}] at (11, 3) { };
\draw (0, 3) -- (10, 3);
\node[proc ,b] at (0, 3) {};
\node[proc ,n] at (10, 3) {};

%\node[label={right:$C_1$}] at (11, 2) { };
\node[label={right:{\color{red}$\Sigma_1$} $\equiv\{\lblanc\}$}] at (11, 2) { };
\draw (0, 2) -- (3.3333333333333335, 2.0);
\draw (3.3333333333333335, 2.0) -- (6.666666666666666, 2.0);
\draw[very thick, red] (6.666666666666666, 2.0) -- (10, 2);
\node[proc ,b] at (0, 2) {};
\node[proc ,n] at (10, 2) {};
\node[proc ,b] at (6.666666666666666, 2.0) {};
\node[proc ,n] at (3.3333333333333335, 2.0) {};

%\node[label={right:$C_2$}] at (11, 1) { };
\node[label={right:{\color{red}$\Sigma_2$} $\equiv\{\lok\lblanc,\lnoir\lok\}$}] at (11, 1) { };
\draw (0, 1) -- (1.1111111111111112, 1.0);
\draw[very thick, red] (1.1111111111111112, 1.0) -- (2.2222222222222223, 1.0);
\draw (2.2222222222222223, 1.0) -- (3.3333333333333335, 1.0);
\draw[very thick, red] (3.3333333333333335, 1.0) -- (4.444444444444445, 1.0);
\draw (4.444444444444445, 1.0) -- (5.555555555555555, 1.0);
\draw (5.555555555555555, 1.0) -- (6.666666666666666, 1.0);
\draw[very thick] (6.666666666666666, 1.0) -- (10, 1.0);
\node[proc ,b] at (0, 1) {};
\node[proc ,n] at (10, 1) {};
\node[proc ,b] at (6.666666666666666, 1.0) {};
\node[proc ,n] at (3.3333333333333335, 1.0) {};
\node[proc1,b] at (2.2222222222222223, 1.0) {};
\node[proc1,n] at (1.1111111111111111, 1.0) {};
\node[proc1,n] at (5.555555555555555, 1.0) {};
\node[proc1,b] at (4.444444444444445, 1.0) {};

\node[label={right:{\color{red}$\Sigma_3$} is all remaining simplices }] at (11, 0) { };
\draw[very thick, red] (0, 0) -- (0.37037037037037046, 0.0);
\draw[very thick, red] (0.37037037037037046, 0.0) -- (0.7407407407407408, 0.0);
\draw[very thick, red] (0.7407407407407408, 0.0) -- (1.1111111111111112, 0.0);
\draw[very thick] (1.1111111111111112, 0.0) -- (2.2222222222222223, 0.0);
\draw[very thick, red] (2.2222222222222223, 0.0) -- (2.5925925925925926, 0.0);
\draw[very thick, red] (2.5925925925925926, 0.0) -- (2.962962962962963, 0.0);
\draw[very thick, red] (2.962962962962963, 0.0) -- (3.3333333333333335, 0.0);
\draw[very thick] (3.3333333333333335, 0.0) -- (4.444444444444445, 0.0);
\draw[very thick, red] (4.444444444444445, 0.0) -- (4.814814814814815, 0.0);
\draw[very thick, red] (4.814814814814815, 0.0) -- (5.185185185185185, 0.0);
\draw[very thick, red] (5.185185185185185, 0.0) -- (5.555555555555555, 0.0);
\draw[very thick, red] (5.555555555555555, 0.0) -- (5.925925925925926, 0.0);
\draw[very thick, red] (5.925925925925926, 0.0) -- (6.296296296296296, 0.0);
\draw[very thick, red] (6.296296296296296, 0.0) -- (6.666666666666666, 0.0);
\draw[very thick] (6.666666666666666, 0.0) -- (10, 0.0);
\node[proc ,b] at (0, 0) {};
\node[proc ,n] at (10, 0) {};
\node[proc ,b] at (6.666666666666666, 0.0) {};
\node[proc ,n] at (3.3333333333333335, 0.0) {};
\node[proc1,b] at (2.2222222222222223, 0.0) {};
\node[proc1,n] at (1.1111111111111111, 0.0) {};
\node[proc2,b] at (0.7407407407407408, 0.0) {};
\node[proc2,n] at (0.37037037037037046, 0.0) {};
% \node[proc2,n] at (1.8518518518518519, 0.0) {};
% \node[proc2,b] at (1.4814814814814816, 0.0) {};
\node[proc2,b] at (2.962962962962963, 0.0) {};
\node[proc2,n] at (2.5925925925925926, 0.0) {};
\node[proc1,n] at (5.555555555555555, 0.0) {};
\node[proc1,b] at (4.444444444444445, 0.0) {};
% \node[proc2,n] at (4.074074074074074, 0.0) {};
% \node[proc2,b] at (3.7037037037037037, 0.0) {};
\node[proc2,b] at (5.185185185185185, 0.0) {};
\node[proc2,n] at (4.814814814814815, 0.0) {};
\node[proc2,n] at (6.296296296296296, 0.0) {};
\node[proc2,b] at (5.925925925925926, 0.0) {};
% \node[proc1,b] at (8.88888888888889, 0.0) {};
% \node[proc1,n] at (7.777777777777777, 0.0) {};
% \node[proc2,b] at (7.407407407407407, 0.0) {};
% \node[proc2,n] at (7.037037037037036, 0.0) {};
% \node[proc2,n] at (8.518518518518517, 0.0) {};
% \node[proc2,b] at (8.148148148148149, 0.0) {};
% \node[proc2,b] at (9.62962962962963, 0.0) {};
% \node[proc2,n] at (9.259259259259258, 0.0) {};

\end{tikzpicture}
  \caption{  \label{fig:ex_sub_ter}Example of a terminating subdivision.\\
  The bolder simplices show the stable edges (colored red at first appearance). For the correspondence of simplices with executions, see section \ref{sub:topo_rep_of_dist_syst}}
\end{figure}

There is an example in Fig.~\ref{fig:ex_sub_ter} of a terminating subdivision for $\Gamma^\omega$.

The following definition expresses when a terminating subdivision
covers all considered scenarios of a given $L$.  The intuition is that
every scenario of $L$ should eventually land into a simplex of $\TS$.
In terms of words, this means that all words of $L$ have a prefix in
the corresponding words of the simplexes of $\TS$.

\begin{definition}
A terminating subdivision $\TS$ of C is \emph{admissible} for 
$L\subseteq\Go$ if for any scenario $\rho\in L$ the corresponding
sequence of edges $\sigma_0,\sigma_1,\ldots$ is such that there exists
$r>0$, $|\sigma_r|$ is a stable edge in $K(\TS)$.
\end{definition}

\label{sub:topo_carac}

We can now state and prove our new characterization theorem. First, for
simplicity, notice that the output complex Consensus
$O_{2gen}$ has two connected components representing 0 and 1. Thus we can
identify it to the set $\{0,1\}$ and define the relation
$\Delta'_{2gen}\subset\I_{2gen}\times\{0,1\}$ analogous to $\Delta_{2gen}$.

\begin{theorem}
  The task $T_{2gen}=(\I_{2gen},\O_{2gen},\Delta_{2gen})$ is solvable in a
  sub-model $L\subseteq\Go$ if and only if there exist a terminating
  subdivision $\Phi$ of $\I_{2gen}$ and a simplicial function
  $\delta\colon K(\Phi)\to\{0,1\}$ such that~:
  \begin{theoenum}
    \item $\Phi$ is admissible for $L$;
    \item \label{thii}
      For all simplex $\sigma\in\I_{2gen}$, if $\tau\in K(\Phi)$ is such
    that $|\tau|\subset|\sigma|$, then $\delta(\tau)\in\Delta'_{2gen}(\sigma)$;
    \item $|\delta|$ is continuous.
  \end{theoenum}
  \label{thm:GACT'2gen}
\end{theorem}

\begin{proof}
  \textbf{Necessary condition.}
  Suppose we have an algorithm $\mathscr{A}$ for solving the Binary
  Consensus task, we will
  construct a terminating subdivision $\Phi$ admissible for $L$ and a function
  $\delta$ that satisfies the conditions of the theorem.

  When running $\mathscr{A}$, we can establish which nodes of the protocol
  complex decide a value by considering the associated state of the process.
  Intuitively, $\Phi$ and $\delta$ are built as the following~:
  in each level $r$, we consider all runs of length $r$ in $L$.
  Each run yields a simplex, if the two nodes of the simplex have 
  decided a (unique) value $v$, we add this simplex to $\Sigma_{r}$ and set
  $\delta(\sigma)=v$.

  Formally, let $\mathscr{C}^L(r)$ be the protocol complex of a full
  information protocol subject to the language $L$ at round $r$ and
  $V(\mathscr{C}^L(r))$ its set of vertices.
  Given a vertex $x$,
  let $val(x)$ be the value decided by the corresponding process in
  the execution that leads to state $x$.  And $\forall r\geq0$ define
  \begin{align*}
    \Sigma_{r}&=\{\{x,y\}\in \mathscr{C}^L(r)\mid
    \text{$x$ and $y$ have both decided and at least one has just
      decided
    in round }r\}\\ 
    \delta(x)&=val(x)\quad\forall x\in V(\Sigma_r)\\
  \end{align*}

  The function $\delta$ is well defined since it depends only on the
  state encoded in $x$ and in $\Sigma_r$ all vertex have decided.
  
  For all $\{x,y\}\in\Sigma_r$, note that $val(x)=val(y)$ because $\mathscr{A}$
  satisfies the \emph{agreement} property.

  By construction for $\Phi$, it is a terminating subdivision.
  Furthermore, $\Phi$ is admissible for $L$.
  Since $\mathscr{A}$ terminates subject to $L$,
  all processes decide a value in a run of $L$, so all
  nodes of the complex protocol restricted to $L$ will be in $\Phi$
  (by one of their adjacent simplexes).

  The condition~\ref{thii} is also satisfied by
  construction because $\mathscr{A}$ satisfies the \emph{validity} property,
  \ie $val(x)\in\Delta'_{2gen}(\sigma)$.

  We still have to prove that $|\delta|$ is continuous, in other words for all
  $x\in|K(\Phi)|$, we must have
  $$\forall\varepsilon>0\quad\exists\eta_x>0\quad\forall y\in|K(\Phi)|\quad
  |x-y|\leq\eta_x\Rightarrow ||\delta|(x)-|\delta|(y)|<\varepsilon$$
  This property for continuity says that when two points $x,y\in|K(\Phi)|$ are
  close, their value of $|\delta|$ is close. When $\varepsilon$ is small (\eg
  $\varepsilon<1/2$), $||\delta|(x)-|\delta|(y)|<\varepsilon$ implies that
  $|\delta|(x)=|\delta|(y)$. This is because $|\delta|$'s co-domain is discrete.
  When $\varepsilon$ is large, the property is always verified and thus is not
  much of interest. We can thus reformulate it without considering
  $\varepsilon$~: we must show that $\forall x\in|K(\Phi)|$

  \begin{eqnarray}\label{continu}
  \exists\eta_x>0\quad\forall y\in|K(\Phi)|\quad |x-y|\leq\eta_x\Rightarrow
  |\delta|(x)=|\delta|(y)
  \end{eqnarray}
  
  In defining $\eta_x$ as followed, we have the continuity condition for all
  $x\in V(K(\Phi))$~:

  $$\eta_x=\min\{\frac{1}{3^{r+1}}\mid \exists r\in\N, \exists y,\{x,y\}
  \in V(\Sigma_r)\}.$$

  Since there exists at least one ($\algo$ terminates for any
  execution in $L$) and at most two such $y$, the minimum is well
  defined for all $x\in V(K(\Phi))$. Moreover, we remark that since
  the geometric realization of simplices of $\Sigma_r$ are of size
  $\frac{1}{3^r}$, the ball centered in $x$ and of diameter $\eta_x$
  is included in the stable simplices and the function $\delta$ is
  constant on this ball (by agreement property).
  
  Let $x,y\in V(K(\Phi))$, let $z\in[x,y]$, we define
  $\eta_z=\min(d(x,z),d(y,z))$.
  By construction, using such function $\eta$,
  for all $z\in|K(\Phi)|$, Proposition (\ref{continu}) is satisfied. 

  \bigskip

  \textbf{Sufficient condition.}
  Given a terminating subdivision $\Phi$ admissible for $\I_{2gen}$ and a
  function $\delta$ that satisfies the conditions of the theorem, we present an
  algorithm $\mathscr{A}$ that solves Consensus.

  First, we describe how to obtain a function $\eta$ from the
  continuity of $|\delta|$.  For any $x\in V(K(\Phi))$, there exists
  $\eta(x)$ such that for any $y\in|K(\TS)|$, when $|x-y|\leq\eta(x)$, we
  have $\delta(y)=\delta(x)$.
  Notice we can choose $\eta$ such that $\forall x\in|\I| \forall y_1,y_2\in|K(\TS)|,
  |y_1-x|\leq\eta(y_1)\wedge|y_2-x|\leq\eta(y_2) \Rightarrow |\delta|(y_1)=|\delta|(y_2).$
  Now, we show how to extend the definition to
  $|K(\Phi)|$.  Consider $w\in |\Phi|\backslash V(K(\Phi))$ and denote $x,y$
  the vertices of $K(\Phi)$ that defines the segment to which $w$
  belongs. We define $\eta(w)=\min\{\eta(x),\eta(y)\}$.

  \bigskip
  We recall that $B(z,t)$ (resp. $\bar{B}(z,t)$) is the open (resp. closed) ball of center $z$ and radius $t$.
  We define a boolean function $Finished(r,x)$ for $r\in\N$ and $x\in |\I|$ that is true when
  $\exists y\in V(Chr^r(\I)), y\in |K_r(\Phi)|, \bar{B}(x,\frac{1}{3^r})\subset B(y,\eta(y))$. 
  
  The Consensus algorithm is described in
  Algorithm~\ref{alg:algo_cons_bin}, using the function $Finished$ we just
  defined from $\eta$.
  Notice that $\eta$ is fully defined on $|K(\Phi)|$ and that the existential condition at line \ref{algline:if} is over a finite subset of $|K(\Phi)|$.

As in Algorithm~\ref{consalgo},
it has messages always of the same type. They
have two components, the first one is the initial bit, named
$init$. The second is an integer named $ind$. Given a message $msg$,
we note $msg.init$ (resp. $msg.ind$) the first (resp. the second)
component of the message. The computation of the index is similar.
We maintain in an auxiliary variable $r$ the current round (line \ref{diamsub}).

The halting condition is now based upon the position of the geometric
realization $geo$ of the current round with regards to the terminating
subdivision $K(\Phi)$ and the $\eta$ function : we wait until the
realization is close enough of a vertex in $|K(\Phi)|$ and the
corresponding open ball of radius $\eta$ contains a neighbourhood of
the current simplex.  Note that when $initb$ or $initw$ is still
$null$, this means that we are at the corners of the square \I.

  \begin{algorithm}
  \KwData{function $\eta$}
  \KwIn{$init\in\{0,1\}$}

  $r=0$\;
  \eIf{$\proc=\noir$}{ind=1\;
  initw=null\;initb=init\;}{ind=0\;initw=init\;initb=null\;}
  \Repeat{$\exists y\in V(Chr^r(\I)), y\in |K_r(\Phi)|, \bar{B}(geo(ind/3^r,initw,initb),\frac{1}{3^r})\subset B(y,\eta(y))$}{\label{algline:if}
    msg = (init,ind)\;
    send(msg)\;
    msg = receive()\;
    \eIf(// message was lost){msg == null}
        {$ind = 3*ind$\;}
        {$ind = 2*msg.ind+ind$\;
          \eIf{$\proc=\noir$}{$initw = msg.init$\;}{$initb = msg.init$\;}}
    $r=r+1$\;
  }
  \KwOut{$|\delta|(y)$\label{diamsub}}
  \caption{Algorithm $\mathscr{A}_{\eta}$ for the binary consensus with two
    processes for process $p$ where $\eta$ is a function $[0,1]\to[0,1]$.
    geo(x,w,b) is the embedding function into the geometric realization.
  \label{alg:algo_cons_bin}
  }
  \end{algorithm}
  
  We show that the algorithm solves Consensus. Consider an execution
  $w$, we will prove it terminates.  The fact that the output does not
  depends on $y$ comes from the choice of $\eta$.  The properties
  Agreement and Validity then come immediately from condition
  \ref{thii} on $\delta$ and $\Delta'_{2gen}$.

  First we prove termination for at least one process. Assume none halts.

  The admissibility of $\Phi$ proves that there is a round $r_0$ for
  which the simplex corresponding to the current partial scenario
  $w_{\mid r_0}$ is a simplex $\sigma$ of $\Phi$. Denote $r_1\geq r_0$ an
  integer such that $\frac{1}{3^{r_1}}<\eta(y)$ with $y\in|\sigma|$.
  Since from round $r_0$, all future geometric realizations will remain in
  $|\sigma|$, we have, at round $r_1$, that $Finished$ is
  true. A contradiction.

  Now, from the moment the first process halts with an index with geometric
  realization $x$ because of some $y$, if the other one has not halted
  at the same moment, then this one will receive no other message from
  the halted process, and its $ind$ variable, with geometric
  realization $z$, will remain constant forever. The closed ball
  centered in $x$ has a neighbourhood inside the open ball centered in
  $y$, therefore there exists $r_2$ such that $\frac{1}{3^{r_2}}$ is
  small enough and $Finished(r_2,z)$ is true.

\end{proof}

We say that the corner of the squares in Fig.~\ref{fig:cpx_proto_cons32} are \emph{corner
  nodes}.  If we compare the two algorithms, we can see that the
differences are in the halting condition and the decision value. In
Algorithm 1, once a forbidden execution $w$ has been chosen, the
algorithm runs until it is far away enough of a prefix of $w$ to
conclude by selecting the initial value of the corner node which is on
the same side as the prefix.

Algorithm 2 is based on the same idea but somehow is more flexible
(more general). If there are many holes in the geometric realization,
it is possible to have different chosen output values for two
different connected components.  Of course, the connected components
that contains the corner nodes have no liberty in choosing the decision value.

\section{About the two Characterizations}
\label{sub:link_with_the_combinatorial_characterization}

In this section, we explain how the combinatorial and the topological
characterizations match. This is of course convenient but the more important
result will be to see that the topological characterization permits to
derive a richer intuition regarding the characterization, in
particular about the status of the special pairs.
We will also illustrate how the Theorem 6.1 of \cite{GKM14} does not
handle correctly the special pairs by showing that in this
cases, combinatorial and geometric simplicial protocol complexes do
differ in precision to handle Consensus computability.

\smallskip

For any adversary $L\subseteq\Go$, $L$ can be either an
obstruction or  solvable for Consensus. In other words, the set of
sub-models $\mathcal{P}(\Go)$ is partitioned in two subset~: obstructions and
solvable languages. Theorem~\ref{thm:FG11} explicitly describes the
solvable languages and classes them in four families. In contrast,
Theorem~\ref{thm:GACT'2gen} gives necessary and sufficient conditions for a
sub-model to be either an obstruction or not.

We first give an equivalent version of Theorem~\ref{thm:GACT'2gen}.

\begin{theorem}
  The task $T_{2gen}$ is solvable in $L\subseteq\Go$ if and only if
  $|\mathcal C^L|$ is not connected.
  \label{thm:cons_ssi_k_deco}
\end{theorem}

\begin{proof}

  ($\Rightarrow$) If $T_{2gen}$ is solvable in $L$, let $\Phi$ and $\delta$
  as described in Theorem~\ref{thm:GACT'2gen}.
  From admissibility, we have that $|\mathcal C^L|\subset|K(\Phi)|.$
  Now $|\delta|$ is a continuous surjective function from
  $|\mathcal C^L|$ to $\{0,1\}$, this implies that the domain $|\mathcal C^L|$ is not
  connected since the image of a connected space by a continuous
  function is always connected.

  ($\Leftarrow$) With $|\mathcal C^L|$ disconnected, we can associate
  output value to connected components in such a way that $\Delta'_{2gen}$ is
  satisfied. Consider the segment $[0,1]$, there exists $z\in[0,1]$
  and $z\notin |\mathcal C^L|$.

  We define $\Phi$ in the following way :
  $$\Sigma_k=\{S\mid
  S=[\overline{ind}(w),\overline{ind}(w)+\frac{1}{3^k}], w\in\Gamma^k,
  \exists w'\in \Gamma^\omega, ww'\in L
  \mbox{ and } \forall i_w,i_b\in\{0,1\}\;geo(z,i_w,i_b)\notin |S|\}$$

  We now denote $\delta$ the function $K(\Phi)\to \{0,1\}$ such that
  for $v\in[0,1]$, for $X=geo(v,i_w,i_b)$, $\delta(X)=i_w$ if $v<z$
  and $\delta(v)=i_b$ otherwise.

  We now check that these satisfies the conditions of
  Theorem~\ref{thm:GACT'2gen}.  Admissibility comes from the fact
  that, by construction of $z$, there are no runs in $L$ that converge
  to $z$, therefore for any run $w\in L$, at some point,
  $\frac{1}{3^k}$ is small enough such that there is an edge in
  $\Sigma_k$ that contains $\overline{ind_{w_{|k}}}$ and not $z$.
  
  The function $|\delta|$ is continuous and has been defined such that condition~\ref{thii} is also clearly satisfied.

\end{proof}

\bigskip

We recall the definitions of $Fair$ and $SPair$ languages, and the admissible language
families defined in Section~\ref{sub:comb_charac}.
\begin{align*}
  Fair(\Go)&=\Go\backslash\{xy\mid x\in\Gamma^*\quad
  y\in\{\lall,\lblanc\}^\omega\cup\{\lall,\lnoir\}^\omega\} \\
  SPair(\Go)&=\{(w,w')\in\Go\times\Go \mid w\neq w', \quad \forall r\in\N \quad
  |ind(w_{|r})-ind(w'_{|r})|\leq1\}
\end{align*}

A message adversary is solvable if it is one of the following (non-exclusive)
forms
\begin{enumerate}
  \item $\F_1 = \{L \mid \exists f\in Fair(\Go), f\notin L\}$
  \item $\F_2 = \{L \mid \exists (w,w')\in SPair(\Go), w,w'\notin L\}$
  \item $\F_3 = \{L \mid \lnoir^\omega\notin L\}$
  \item $\F_4 = \{L \mid \lblanc^\omega\notin L\}$
\end{enumerate}

We give now a topological interpretation of this description.
We start with a topological characterization of special pairs.

\begin{lemma}
  Let $w,w'\in \Gamma^\omega$, $w\neq w'$,
  $\overline{ind}(w)=\overline{ind}(w')$ if and only $(w,w')\in SPair(\Go)$.
\end{lemma}

\begin{proof}
  This comes from Lemma~\ref{diff2} and the very definition of special pairs.
\end{proof}

In others words, special pairs are exactly the runs that have another
run that converges to the same geometric realization.
So removing only one member of a special pair is not enough
to disconnect $|\mathcal C^L|$.
It is also straightforward to see that removing a fair run implies
disconnection and that removing $\lnoir^\omega$ or $\lblanc^\omega$
implies disconnecting at the corners.

\subsection{Counter-Example}\label{contrex}
Now we look at Theorem~6.1 from \cite{GKM14}. The only difference is that in \cite{GKM14} $\delta$ is only required to be chromatic. Here we have shown that we need a stronger condition that is the
continuity of $|\delta|$. We explain why this is strictly stronger, ie  how it is possible to have
simplicial mapping while not having continuity of the geometric
realization.

Consider $L=\Gamma^\omega\backslash w_0$, with
$w_0=\lok\lblanc^\omega$. We define $\Phi$ as follows.

$\Sigma_1=\{[0,\frac{1}{3}],[\frac{2}{3},1]\}$
For $r\geq 2$, $\Sigma_r=\{[\frac{2}{3}-\frac{1}{3^{r-1}},\frac{2}{3}-\frac{2}{3^{r}}],[\frac{2}{3}-\frac{2}{3^{r}},\frac{2}{3}-\frac{1}{3^{r}}]\}$.

The terminating subdivision $\Phi$ is admissible for $L$ and there
exists a simplicial mapping from $K(\Phi)$ to $\{0,1\}$. We can set
$\delta(x)=0$ when $x<\frac{2}{3}$ and $\delta(x)=1$ otherwise.

$K(\Phi)$ has two connected components $[\frac{2}{3},1]$ on one side,
and all the other segments on the other side. The function $\delta$ is
therefore simplicial since there is no $[z,\frac{2}{3}]$ interval in
$K(\Phi)$. 

So such a $L$ satisfies the assumptions for Theorem~6.1 of
\cite{GKM14}.  To see that it does not satisfy the assumptions of
Theorem~\ref{thm:GACT'2gen}, we remark that $\mid K(\Phi)\mid$ is
$[0,1]$ and the $\delta$ function we define above does not have a geometric realization that is continuous
and moreover there is no way to define a continuous surjective
function from $[0,1]$ to $\{0,1\}$. So the statements of the Theorem
are not equivalent.  To correct Theorem~6.1 of \cite{GKM14}, it is
needed to add an assumption, that is that $L$ has to be ``closed'' for
special pairs : either both members belong to $L$ or none, this has
been confirmed by the authors \cite{K15}.

The simplicial complex $K(\Phi)$ has two connected components when
seen as an abstract simplicial complex, however, it is clear that the
geometric realization of $K(\Phi)$ is the entire interval $[0,1]$, it
has one connected component.

\section{Conclusion}
To conclude, we have that the two Theorems~\ref{thm:FG11} and
\ref{thm:GACT'2gen} are indeed equivalent, even if their formulation
is very different.  We emphasize that the topological characterization
with Theorem~\ref{thm:cons_ssi_k_deco} gives a better explanation of
the results primarily obtained in \cite{FG11}.  The different cases of
Theorems~\ref{thm:FG11} are unified when considered topologically.
Note also that in the general case, the topological reasoning should
be done on the continuous version of simplicial complexes, not on the
abstract simplicial complexes.  From Section \ref{contrex}, we see we
have a simplicial complex $K(\Phi)$ that is disconnected when seen as
an abstract simplicial complex, but whose embedding (geometric
realization) makes for a connected space.

This study of the solution to the Consensus problem in the general
case for two processes is another argument in favor of topological
methods in Distributed Computability.

We are aware of \cite{consensus-epistemo} that appeared between
revisions of this paper. We underline that Theorem~\ref{thm:GACT'2gen}
cannot be obtained in a straightforward way from
\cite{consensus-epistemo}.

\end{document}